\documentclass[draftclsnofoot, onecolumn, 11pt]{IEEEtran}

\usepackage{caption}
\usepackage{subcaption}
\usepackage{mathrsfs}
\usepackage[cmex10]{amsmath}
\usepackage{cite,graphicx,epsfig,wrapfig,amsfonts,amssymb,algorithmic,algorithm,threeparttable,color,url}
\usepackage{mdwmath,multirow}
\usepackage{mdwtab}
\usepackage{float}
\usepackage{amsthm}

\newtheorem{theorem}{Theorem}
\newtheorem{lemma}{Lemma}

\newtheorem{corollary}{Corollary}
\theoremstyle{definition}

\graphicspath{{./}{figures/}}

\ifCLASSINFOpdf

\else

\fi

\begin{document}

\begin{titlepage}
\begin{center}
\vspace*{-2\baselineskip}
\begin{minipage}[l]{7cm}
\flushleft
\includegraphics[width=2 in]{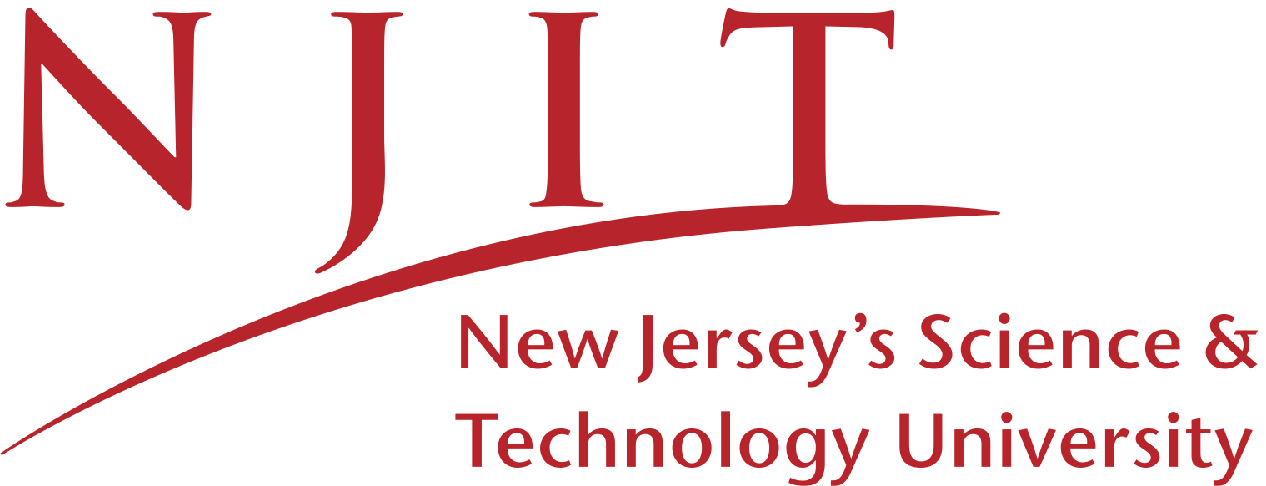}
\end{minipage}
\hfill
\begin{minipage}[r]{7cm}
\flushright
\includegraphics[width=1 in]{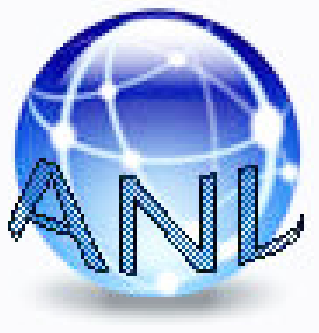}%
\end{minipage}

\vfill

\textsc{\LARGE Network Utility Aware Traffic Loading Balancing in Backhaul-constrained Cache-enabled Small Cell Networks with Hybrid Power Supplies}\\

\vfill
\textsc{\LARGE Tao Han\\[12pt]
\LARGE NIRWAN ANSARI}\\
\vfill
\textsc{\LARGE TR-ANL-2014-007\\[12pt]
\LARGE Sep. 29, 2014}\\[1.5cm]
\vfill
{ADVANCED NETWORKING LABORATORY\\
 DEPARTMENT OF ELECTRICAL AND COMPUTER ENGINEERING\\
 NEW JERSY INSTITUTE OF TECHNOLOGY}
\end{center}
\end{titlepage}

\title{Network Utility Aware Traffic Loading Balancing in Backhaul-constrained Cache-enabled Small Cell Networks with Hybrid Power Supplies}
\author{\IEEEauthorblockN{Tao Han, \emph{{Student Member, IEEE}}, and
Nirwan Ansari, \emph{{Fellow, IEEE}}}\\
\IEEEauthorblockA{Advanced Networking Laboratory \\
Department of Electrical and Computer Engineering \\
New Jersey Institute of Technology, Newark, NJ, 07102, USA\\
Email:  \{th36, nirwan.ansari\}@njit.edu}

\thanks{This work was supported in part by NSF under grant no. CNS-1218181 and no. CNS-1320468.}}
\maketitle
\pagestyle{headings}
\thispagestyle{empty}
\begin{abstract}
Explosive data traffic growth leads to a continuous surge in capacity demands across mobile networks. In order to provision high network capacity, small cell base stations (SCBSs) are widely deployed. Owing to the close proximity to mobile users, SCBSs can effectively enhance the network capacity and offloading traffic load from macro BSs (MBSs). However, the cost-effective backhaul may not be readily available for SCBSs, thus leading to backhaul constraints in small cell networks (SCNs). Enabling cache in BSs may mitigate the backhaul constraints in SCNs. Moreover, the dense deployment of SCBSs may incur excessive energy consumption. To alleviate brown power consumption, renewable energy will be explored to power BSs. In such a network, it is challenging to dynamically balance traffic load among BSs to optimize the network utilities. In this paper, we investigate the traffic load balancing in backhaul-constrained cache-enabled small cell networks powered by hybrid energy sources. We have proposed a network utility aware (NUA) traffic load balancing scheme that optimizes user association to strike a tradeoff between the green power utilization and the traffic delivery latency. On balancing the traffic load, the proposed NUA traffic load balancing scheme considers the green power utilization, the traffic delivery latency in both BSs and their backhaul, and the cache hit ratio. The NUA traffic load balancing scheme allows dynamically adjusting the tradeoff between the green power utilization and the traffic delivery latency. We have proved the convergence and the optimality of the proposed NUA traffic load balancing scheme. Through extensive simulations, we have compared performance of the NUA traffic load balancing scheme with other schemes and showed its advantages in backhaul-constrained cache-enabled small cell networks with hybrid power supplies.
\end{abstract}
\IEEEpeerreviewmaketitle
\section{Introduction}
\label{sec:introduction}
Owing to the proliferation of mobile devices and bandwidth greedy applications, mobile data traffic grows exponentially that has led to a continuous surge in network capacity demands~\cite{Han:2013:OAC}. Small cell base stations (SCBSs) are widely deployed to provision high network capacity~\cite{Andrews:2014:AOLB}. SCBSs, with a small coverage area, can significantly improve the spectrum utilization in mobile networks and thus increase the network capacity \cite{Nakamura:2013:TSC}. However, owing to the disparate transmit powers and base station (BSs) capabilities, traditional traffic load balancing metrics such as the signal-to-interference-plus-noise ratio (SINR) and the received-signal-strength-indication (RSSI) may lead to a severe traffic load imbalance \cite{Andrews:2014:AOLB}. Hence, in order to fully exploit the capacity potential of small cell networks~(SCNs), the traffic load balancing scheme should be well designed.

In mobile networks, traffic load balancing is achieved by executing user association process in which mobile users are assigned to base stations (BSs) for services. Various user association algorithms have been proposed to optimize the traffic load among BSs ~\cite{Andrews:2014:AOLB,Jo:2012:HCN,Ye:2013:UAL,Kim:2012:DOU,Aryafar:2013:RSG}.
Most of the existing solutions optimize the traffic load balancing in a mobile network with the implication that the air interface between BSs and mobile users is the bottleneck of the network. This implication is generally correct for BSs whose deployments are well planned. However, considering the potentially dense deployment of SCBSs, various backhaul solutions, e.g., xDSL, non-line-of-sight (NLOS) microwave, wireless mesh networks, rather than ideal backhaul such as optical fiber and LOS microwave, may be adopted~\cite{Nakamura:2013:TSC}. As a result, backhaul instead of BSs may become the bottleneck of SCNs. To alleviate the backhaul constraints, content caching techniques have been explored to enable caching popular contents in BSs to reduce the traffic load in backhaul \cite{Wang:2014:CIT,Poularakis:2014:AAM,Shanmugam:2013:FWC,Monserrat:2014:RMW}.
Therefore, it is desired to optimize the user association with the consideration of backhaul constraints and the performance of BSs' content cache system in SCNs.

Enhancing energy efficiency is also a critical task for next generation mobile networks \cite{Han:2012:OGC,Hasan:2011:GCN}. Although SCBSs consume less power than macro BSs (MBSs), the number of SCBSs will be orders of magnitude larger than that of MBSs for a wide scale network deployment. Hence, the overall power consumption of SCNs will be phenomenal. As energy harvesting technologies advance, renewable energy such as sustainable biofuels, solar and wind energy can be utilized to power BSs  \cite{Han:2014:PMN}. Telecommunication companies such as Ericsson and Nokia Siemens have designed renewable energy powered BSs for mobile networks \cite{Ericson:2007:SEU}. Define the electricity pulled from renewable energy systems and the power grid as green power and brown power, respectively.
By adopting renewable energy powered BSs, mobile networks may further reduce their brown power consumption.
However, since the electricity generated from renewable energy is not stable, green power may not be a reliable energy source for mobile networks. Therefore, future SCNs are likely to adopt hybrid energy supplies: brown power and green power. Green power is utilized to reduce the brown power consumption while brown power is utilized as a backup power source \cite{Han:2014:vGALA}. In order to optimize green power utilization, it is desirable to balance the traffic load according to the availability of green power. For instance, mobile networks may enable BSs with sufficient green power to serve more traffic load while reducing the traffic load of BSs consuming brown power \cite{Han:2013:OOG}. Such traffic load balancing strategies, however, may not maximize network utilities such as the network capacity and the traffic delivery latency. Hence, a trade-off between the green power utilization and network utilities should be carefully evaluated in balancing traffic load among BSs.

\begin{figure}[htb]
\centering
\includegraphics[scale=0.25]{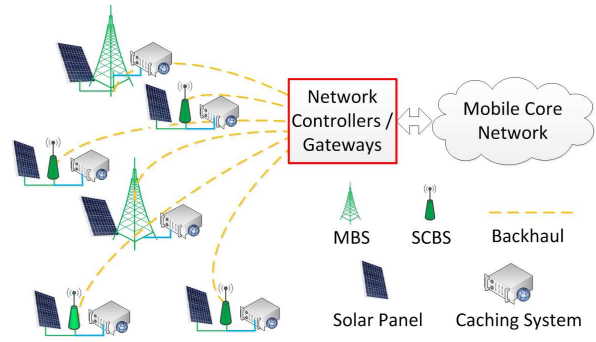}
\caption{The small cell network.}
\label{fig:sc_net_arch}
\end{figure}

In this paper, we investigate the traffic load balancing in backhaul-constrained cache-enabled small cell networks with hybrid power supplies. The network architecture is shown in Fig. \ref{fig:sc_net_arch}. The traffic load balancing in such a network requires the consideration of the green power utilization, BS capacity, backhaul constraints and the performance of cache systems. Therefore, we introduce four network utilities on balancing traffic load among BSs: 1) the green power utilization, 2) the traffic delivery latency in BSs, 3) the traffic delivery latency in backhaul, and 4) the cache hit ratio. The last three network utilities jointly determine the traffic delivery latency of the network. Thus, the awareness of these network utilities helps reduce the traffic delivery latency in the network. Since the green power utilization and the traffic delivery latency are not optimized simultaneously in most scenarios, the tradeoff between the green power utilization and the traffic delivery latency should be determined based on network conditions. We propose the network utility aware (NUA) traffic load balancing scheme to adapt the user association according to the dynamics of these network utilities and strike an adjustable tradeoff between the brown power consumption and the traffic delivery latency. We prove the convergence and the optimality of the proposed NUA traffic load balancing scheme and validate its performance through extensive simulations.

The rest of the paper is organized as follows.
In Section \ref{sec:sys_model}, we define the system model and formulate the traffic load balancing problem. Section \ref{sec:distributed_scheme} presents the proposed NUA traffic load balancing scheme and analyzes its properties. Section \ref{sec:simulation} shows the simulation results, and concluding remarks are presented in Section \ref{sec:conclusion}.

\section{Related Works}
\label{sec:related_work}
Balancing traffic load in mobile networks has been extensively studied in recent years \cite{Wang:2013:MMN,Andrews:2014:AOLB}. In this section, we provide a briefly overview on existing traffic load balancing schemes. The most practical traffic load balancing approach is the cell range expansion (CRE) technique that biases users' receiving signal-to-interference-and-noise-ratios (SINRs) or data rates from some BSs to prioritize these BSs in associating with users \cite{Damn:2011:S3GPP}. Owing to the transmit power difference between MBSs and SCBSs, a large bias is usually given to SCBSs to offload users to small cells \cite{Andrews:2014:AOLB}. By applying CRE, a user associates with the BS from which the user receives the maximum biased SINR or biased data rate \cite{Jo:2012:HCN}. Deriving the optimal bias for BSs is challenging. Singh \emph{et al.} \cite{Singh:2013:OHN} investigated the impact of the bias on network performances and provided a comprehensive analysis on traffic load balancing using CRE in heterogeneous mobile networks.

Optimization theory and game theory have been adopted to solve the traffic load balancing problem. Ye~\emph{et~al.}~\cite{Ye:2013:UAL} modeled the traffic load balancing problem as a utility maximization problem and developed distributed user association algorithms using the primal-dual decomposition. Kim~\emph{et al.}~\cite{Kim:2012:DOU} proposed an $\alpha$-optimal user association algorithm to achieve flow level load balancing under spatially heterogeneous traffic distribution. The proposed algorithm is based on convex optimization theory and may maximize different network utilities by selecting the value of $\alpha$. Aryafar~\emph{et al.}~\cite{Aryafar:2013:RSG} applied game theory to solve the traffic load balancing problem. The authors modeled the problem as a congestion game in which users are the players and user association decisions are the actions. Pantisano~\emph{et al.}~\cite{Pantisano:2014:CUA} formulated the traffic load balancing problem in backhaul constrained SCNs as a one-to-many matching game between SCBSs and users and proposed a distributed algorithm based on the deferred acceptance scheme to obtain a stable matching for mobile users.

Recognizing the green power utilization as one of the performance metrics when balancing the traffic load, Zhou~\emph{et~al.}~\cite{Zhou:2010:ESA:} proposed a handover parameter tuning algorithm for target cell selection, and a power control algorithm for coverage optimization to guide mobile users to access the BSs with renewable energy supply. Han and Ansari \cite{Han:2013:OOG} proposed to optimize the utilization of green power for cellular networks by optimizing BSs' transmit powers. The proposed algorithm achieves significant brown power savings by scheduling the green power consumption along the time domain for individual BSs, and balancing the green power consumption among BSs. The authors have also proposed a user association framework that jointly optimizes the traffic delivery latency and the green power utilization \cite{Han:2013:GALA,Han:2014:vGALA}.

\section{System Model and Problem Formulation}
\label{sec:sys_model}
In this section, we present the system model and the problem formulation. The system model includes the traffic and QoS model, and the energy model.
\subsection{Traffic and QoS model}
\label{subsec:qos_model}
Denote $\mathcal{B}$ as the set of BSs including both MBSs and SCBSs and $\mathcal{A}$ as the coverage area of all BSs. Here, a BS refers to either a MBS or a SCBS. Since BSs are equipped with cache, users' data requests can be fulfilled by the cache system if the requested content is cached; otherwise, the requested content is retrieved from Internet. Retrieving contents from Internet generates traffic loads in a BS's backhaul. Therefore, we model the traffic delivery process as a queuing system as shown in Fig.~\ref{fig:sc_queue_sys}.

\begin{figure}[htb]
\centering
\includegraphics[scale=0.3]{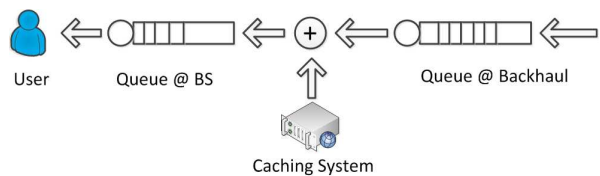}
\caption{The traffic delivery process as a queuing system.}
\label{fig:sc_queue_sys}
\end{figure}

The performance of a cache system is commonly evaluated based on the cache hit ratio that is defined as the ratio between the number of cache hits and the total requests observed over a period of time~\cite{Gomma:2013:EIC}. Many analytical models have been proposed to estimate the hit ratio of a cache system and various content caching strategies have been designed to optimize the performance of cache systems~\cite{Gomma:2013:EIC,Breslau:1999:WCZ,Rodriguez:2001:AWC,Jelenkovic:2003:AIL,Zhang:2012:OWA}.
Thus, in the paper, we assume that the hit ratio of a cache system in a BS can be estimated for the time duration of one user association process and denote $0\leq\alpha_{j}\leq 1$ as the hit ratio of the cache system in BS $j$. Note that how to optimize and estimate the hit ratio is out of the scope of this paper.

Let a mobile user at location~$x$ associate with BS $j$. We assume that the traffic arrives at BS $j$'s backhaul toward the user according to a Poisson process with the arrival rate equaling to $\tilde{\lambda}(x)$, and the traffic loads per arrival (packet sizes per arrival) have an exponential distribution with the average traffic load of $\nu(x)$. We assume that the users associated with BS $j$ are uniformly distributed in its coverage area and the traffic arrival processes are independent. For presentation simplicity, we further assume that no users share the same locations, i.e., only one user at location~$x$. Since the traffic arrival toward a location is a Poisson process, the traffic arrival in BS $j$'s backhaul, which is the sum of the traffic arrivals from its coverage area, is also a Poisson process.
Although BSs may adopt different access technologies as their backhaul, it is reasonable to assume the expected data rates of the backhaul are constant in the time duration of one user association process \cite{Pantisano:2014:CUA}. Since the traffic load per arrival follows an exponential distribution, the traffic delivery time (service time) of the backhaul is also an exponential distribution. Therefore, the traffic delivery in backhaul simply realizes an M/M/1 queuing system.

Denoting $R_{j}$ as the average data rate of BS $j$'s backhaul. To fulfill the traffic demand of the user at location~$x$ , the required service time in BS $j$'s backhaul is
\begin{equation}
\label{eq:bh_required_time}
\tilde{\gamma}(x)= \frac{\nu(x)}{R_{j}}.
\end{equation}
The average traffic load density generated by a user at location~$x$ in BS~$j$'s backhaul is
\begin{equation}
\label{eq:bh_point_load}
\tilde{\varrho}_{j}(x)=\frac{\tilde{\lambda}(x)\nu(x)\eta_{j}(x)}{R_{j}}
\end{equation}
Here, $\eta_{j}(x)=\{0,1\}$ is an indicator function. If $\eta_{j}(x)=1$, the user at location~$x$ is associated with BS~$j$; otherwise, the user is not associated with BS $j$.
Since mobile users are uniformly distributed in the area, the traffic load in BS~$j$'s backhaul can be expressed as
\begin{equation}
\label{eq:backhaul_load}
\tilde{\rho}_{j}=\int_{x \in \mathcal{A}}\tilde{\varrho}_{j}(x)dx.
\end{equation}
According to the properties of the M/M/1 queue \cite{Kleinrock:1976:QS}, the average waiting time for traffic load $\nu(x)$ in BS $j$'s backhaul is
\begin{equation}
\label{eq:bh_ave_waiting_time}
\tilde{W}_{j}(x)=\frac{\tilde{\rho}_{j}\nu(x)}{R_{j}(1-\tilde{\rho}_{j})}.
\end{equation}
Denote $\tilde{\mu}_{j}(x)$ as the latency ratio that measures how much time a user at location~$x$ must be sacrificed in waiting for per unit service time in BS $j$'s backhaul.
\begin{equation}
\label{eq:bh_waiting_per_service}
\tilde{\mu}_{j}(x)=\frac{\tilde{W}_{j}(x)}{\tilde{\gamma}(x)}=\frac{\tilde{\rho}_{j}}{1-\tilde{\rho}_{j}}.
\end{equation}
Since $\tilde{\mu}_{j}(x)$ only depends on the traffic load in BS $j$'s backhaul, all the users associated with BS $j$ have the same latency ratio. Thus, we define
\begin{equation}
\label{eq:bh_latency_ratio}
\tilde{\mu}_{j}(\tilde{\rho_{j}})=\frac{\tilde{\rho_{j}}}{1-\tilde{\rho_{j}}}
\end{equation}
as the latency ratio of BS $j$'s backhaul.
A smaller $\tilde{\mu}_{j}(\tilde{\rho_{j}})$ indicates that BS~$j$'s backhaul introduces less latency to its associated users.

According to Burke's Theorem~\cite{Kleinrock:1976:QS}, the traffic departure process at a BS's backhaul is a Poisson process with average departure rate equaling to the average traffic arrival rate. Therefore, the average traffic arrival rate in BS $j$ toward a user at location~$x$ equals to $\tilde{\lambda}(x)$. Since the hit ratio of BS~$j$'s cache system is $\alpha_{j}$, the data traffic from the backhual accounts for $(1-\alpha_{j})$ of the total traffic load toward the user at location~$x$. Therefore, the average traffic arrival rate in BS~$j$ toward the user at location~$x$ is
\begin{equation}
\lambda(x)=\frac{\tilde{\lambda}(x)}{(1-\alpha_{j})}.
\end{equation}
Since $\alpha_{j}$ is assumed to be a constant during one user association process, the traffic arrival process toward the user at location~$x$ is a Poisson process. In BS $j$, users at different locations may have different data rates depending on channel conditions. When associating with BS $j$, the user's data rate, $r_{j}(x)$, can be generally expressed as a logarithmic function of the perceived SINR, $SINR_{j}(x)$, according to the Shannon-Hartley Theorem~\cite{Kim:2012:DOU},
\begin{equation}
\label{eq:user_rate}
r_{j}(x)=log_{2}(1+SINR_{j}(x)).
\end{equation}
Here,
\begin{equation}
\label{eq:user_SINR}
SINR_{j}(x)=\frac{P_{j}g_{j}(x)}{\sigma^{2}+\sum_{k \in \mathcal{B},k\neq j}P_{k}g_{k}(x)}.
\end{equation}
Here, $P_{j}$ is the transmission power of BS~$j$, and $\sigma^{2}$ denotes the noise power level. Since the users' data rate is generally distributed, the service time in BS~$j$ follows a general distribution. Therefore, a BS's downlink transmission process realizes a M/G/1 processor sharing (PS) queue, in which multiple users share the BS's downlink radio resource~\cite{Kleinrock:1976:QS}.

In mobile networks, various downlink scheduling algorithms have been proposed to enable proper sharing of the limited radio resource in a BS. According to the scheduling algorithm, users may be assigned different priorities on sharing the radio resource. For analytical simplicity, we assume that mobile users are served based on the round robin (RR) fashion. Then, the average traffic load density at location~$x$ in BS~$j$ is calculated as
\begin{equation}
\label{eq:bs_point_load}
\varrho_{j}(x)=\frac{\lambda(x)\nu(x)\eta_{j}(x)}{r_{j}(x)}
\end{equation}
The traffic load in BS~$j$ can be expressed as
\begin{equation}
\label{eq:bs_load}
\rho_{j}=\int_{x \in \mathcal{A}}\varrho_{j}(x)dx.
\end{equation}
This value of $\rho_{j}$ indicates the fraction of time during which BS~$j$ is busy.
To fulfill the traffic demand of a user located at~$x$, the required service time in BS~$j$ is
\begin{equation}
\label{eq:required_time}
\gamma(x)= \frac{\nu(x)}{r_{j}(x)}.
\end{equation}
Since the traffic delivery process in a BS realizes a M/G/1-RR queue, the average traffic delivery time for the user in BS~$j$~\cite{Kleinrock:1976:QS} is
\begin{equation}
\label{eq:traffic_deliver_time}
T_{j}(x)=\frac{\nu(x)}{r_{j}(x)(1-\rho_{j})}.
\end{equation}
The average waiting time for traffic load $\nu(x)$ in BS~$j$ is
\begin{equation}
\label{eq:ave_waiting_time}
W_{j}(x)=T_{j}(x)-\gamma(x)=\frac{\rho_{j}\nu(x)}{r_{j}(x)(1-\rho_{j})}.
\end{equation}
Denote $\mu_{j}(x)$ as the latency ratio of BS~$j$ for a user at location~$x$.
\begin{equation}
\label{eq:waiting_per_service}
\mu_{j}(x)=\frac{W_{j}(x)}{\gamma(x)}=\frac{\rho_{j}}{1-\rho_{j}}.
\end{equation}
$\mu_{j}(x)$ only depends on the traffic load in BS $j$. Therefore, all the users associated with BS~$j$ have the same latency ratio. Thus, we define
\begin{equation}
\label{eq:latency_ratio}
\mu_{j}(\rho_{j})=\frac{\rho_{j}}{1-\rho_{j}}
\end{equation}
as the latency ratio of BS $j$. A smaller $\mu_{j}(\rho_{j})$ indicates that BS~$j$ introduces less latency to its associated users. In this paper, we use $\mu_{j}$ and $\tilde{\mu}_{j}$ to reflect the traffic delivery latency in BS~$j$ and its backhual, respectively, we adopt $\mu_{j}(\rho_{j})+\tilde{\mu}_{j}(\tilde{\rho}_{j})$ as the QoS model that indicates the latency of delivering traffic through BS~$j$.

\subsection{Energy model}
\begin{figure}
\centering
\includegraphics[scale=0.3]{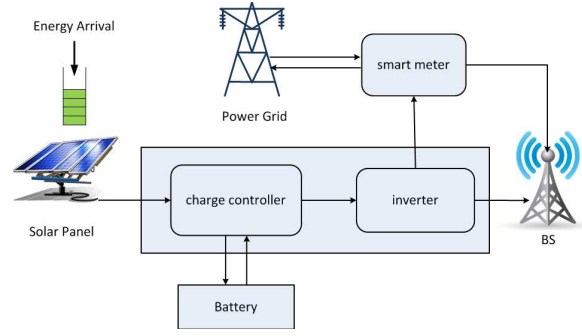}
\caption{A hybrid energy powered BS.}
\label{fig:green_bs}
\end{figure}
In the network, BSs have their own renewable energy systems (solar panels) for generating electricity. Meanwhile, BSs are connected with the power grid for electricity supplies. Thus, BSs are powered by hybrid energy sources: green power and brown power. A BS consumes brown power when green power is not sufficient. We assume that the green power systems in MBSs have a higher energy generation capacity than that of SCBSs because MBSs usually consume more energy than SCBSs owing to a relatively large coverage area. Fig.~\ref{fig:green_bs} shows a reference design of a hybrid energy powered BS~\cite{Han:2014:PMN}.

The charge controller optimizes the green power utilization based on the solar power intensity, the power consumption of BSs, and prices of energy drawn from the power grid. Based on the optimization, the charge controller determines how much green power should be utilized to power a BS during a specific time period, e.g., the time duration between two consecutive traffic load balancing procedures \cite{Han:2014:vGALA}. In this paper, we focus on how to balance traffic load among BSs to reduce the traffic delivery latency as well as the brown power consumption within the duration of a traffic balancing procedure. Investigating how to optimize the green power utilization over the time horizon is out of the scope of this paper. Thus, we assume that the amount of available green power for powering a BS within the duration of one user association process is given by the charge controller as a constant \cite{Han:2014:vGALA}. This assumption is reasonable because the traffic load balancing process is at a time scale of several minutes~\cite{Kim:2012:DOU} while solar power generation is usually modeled at a time scale of a hour~\cite{Farbod:2007:RAO}. Denote $e_{j}$ as the amount of green power for powering BS~$j$ in a user association procedure. If BS~$j$'s power consumption is larger than $e_{j}$, the BS consumes brown power. Otherwise, we simply model the BS's brown power consumption are zero\footnote{We do not consider the redistribution of the residual green power in our model, which is out of the scope of this paper.}.

The BS's power consumption includes two parts: the static power consumption and the dynamic power consumption~\cite{Auer:2011:HME}. The static power consumption is the power consumption of a BS without carrying any traffic load. The dynamic power consumption refers to the additional power consumption incurred by traffic load in the BS, which can be well approximated by a linear function of the traffic load \cite{Auer:2011:HME}. Denote $p^{s}_{j}$ as BS~$j$'s static power consumption. Then, BS $j$'s power consumption is
\begin{equation}
\label{eq:bs_power_consumption}
p_{j}=\beta_{j}\rho_{j}+p^{s}_{j}.
\end{equation}
Here, $\beta_{j}$ is the load-power coefficient that reflects the relationship between BS~$j$'s traffic load and its dynamic power consumption. The BS power consumption model can be adjusted to model the power consumption of either MBSs or SCBSs by incorporating and tweaking the static power consumption and the load-power
coefficient \cite{Han:2014:vGALA}. The BS~$j$'s brown power consumption is
\begin{equation}
\label{eq:bs_brown_energy}
p^{b}_{j}=\max{(p_{j}-e_{j},0)}.
\end{equation}
\subsection{Problem formulation}
In determining the user association, the network aims to not only enhance the network QoS by reducing the traffic delivery latency but also reduce the brown power consumption by improving the green power utilization. Owing to the dynamics of the data traffic and green power, the user association that minimizes the traffic delivery latency does not necessarily maximize the green power utilization. Thus, the traffic load balancing problem strives for a trade-off between the traffic delivery latency and the brown power consumption.

According to E.q. (\ref{eq:bs_brown_energy}), brown power is consumed only when green power is not sufficient ($e_{j}<p_{j}$). Given $e_{j}$, the maximum traffic load can be supported by green power in BS~$j$ is
\begin{equation}
\label{eq:bs_green_capacity}
\hat{\rho}_{j}=max(0, min(\frac{e_{j}-p_{j}^{s}}{\beta_{j}},1-\epsilon)).
\end{equation}
Here, $\epsilon$ is an arbitrary small positive constant to guarantee $0\leq\hat{\rho}_{j}<1$. Define $\hat{\rho}_{j}$ as BS~$j$'s green traffic capacity. When BS~$j$'s traffic load is larger than $\hat{\rho}_{j}$, BS~$j$ consumes brown power. In this case, it is desirable to offload data traffic from BS $j$ to alleviate its brown power consumption. Let
\begin{equation}
\label{eq:load_diff}
\bar{\rho}_{j}=\rho_{j}-\hat{\rho}_{j}.
\end{equation}
When $\bar{\rho}_{j}>0$, BS $j$'s traffic load is larger than its green power capacity and thus it is desirable to offload traffic from BS $j$ to save brown power; when $\bar{\rho}_{j}<0$, it is desirable to let BS $j$ carry additional traffic load to enhance the usage of green power and thus reduce other BSs' brown power consumption. However, balancing traffic purely based on the energy consumption may lead to heavy traffic congestion that increases the traffic delivery latency in BSs. In order to strive a balance, we introduce latency weights for individual BSs. Denote
\begin{equation}
\label{eq:latency_weight}
w_{j}(\rho_{j})=e^{\kappa\bar{\rho}_{j}}
\end{equation}
as BS $j$'s latency weight. $\kappa\geq 0$ is a system parameter that adjusts the value of the latency weight.

Aiming to save brown power as well as to reduce the traffic delivery latency of the network, the traffic load balancing problem is formulated as
\begin{eqnarray}
\label{eq:object_network_goal}
\min_{\boldsymbol{\eta}} && \sum_{j \in \mathcal{B}}w_{j}(\rho_{j})(\mu_{j}(\rho_{j})+\tilde{\mu}_{j}(\tilde{\rho}_{j}))\\
\label{eq:constraint_omge}
subject\; to: && \rho_{j}=\int_{x \in \mathcal{A}}\frac{\lambda(x)\nu(x)\eta_{j}(x)}{r_{j}(x)}dx, \nonumber\\
&& \tilde{\rho}_{j}=\int_{x \in \mathcal{A}}\frac{\tilde{\lambda}(x)\nu(x)\eta_{j}(x)}{R_{j}}dx, \nonumber\\
&& 0\leq\rho_{j}\leq 1-\epsilon,  \nonumber\\
&& 0\leq\tilde{\rho}_{j}\leq 1-\epsilon, \nonumber\\
&& \eta_{j}(x)=\{0,1\}, \forall j \in\mathcal{B}, x \in\mathcal{A}.
\end{eqnarray}
Here, $\boldsymbol{\eta}=\{\boldsymbol{\eta}_{j}|j\in\mathcal{B}\}$ and $\boldsymbol{\eta}_{j}=\{\eta_{j}(x)|x\in\mathcal{A}\}$. Based on the formulation, if BS $j$ has sufficient green power ($\tilde{\rho_{j}}\geq\rho_{j}$), ${0<w_{j}(\rho_{j})\leq 1}$; otherwise, $w_{j}(\rho_{j})>1$. A large latency weight grants a BS a high priority in minimizing Eq. (\ref{eq:object_network_goal}) as compared with those of the BSs having a small latency weight. In other words, a large latency weight grants a BS a high priority in offloading traffic. As compared with $w_{j}(\rho_{j})\leq 1$, $w_{j}(\rho_{j})>1$ enables BS $j$ to achieve a smaller latency ratio. Since $\frac{d\mu_{j}(\rho_{j})}{d\rho_{j}}>0$ and $\frac{d\tilde{\mu}_{j}(\tilde{\rho}_{j})}{d\tilde{\rho}_{j}}>0$, a small latency ratio indicates that BS $j$ carries a lighter traffic load, which is desirable for a BS which is consuming brown power ($w_{j}(\rho_{j})\leq 1$). $\kappa$ is a system parameter that enables the network dynamically controlling the trade-off between the brown power consumption and the traffic delivery latency.

\section{Network Utility Aware Traffic Load Balancing}
\label{sec:distributed_scheme}
In this section, we propose the network utility aware (NUA) traffic load balancing scheme and prove its properties. The network utilities considered in the traffic load balancing consist of 1) the green power utilization (brown power consumption), 2) the traffic delivery latency in BSs' backhaul, 3) the traffic delivery latency in BSs, and 4) the hit ratio of BSs' cache systems. The proposed network utility aware traffic load balancing scheme is able to adapt the traffic load among BSs and their backhauls according to the dynamics of these network utilities.

\subsection{Traffic load balancing procedures}
The traffic load balancing procedures can be implemented in either a distribute or a centralized fashion~\cite{Han:2014:vGALA}. For a distributed traffic load balancing, users select their serving BSs based on the operating parameters, e.g., traffic loads and data rates received from BSs. This will incur several interactions between BSs and users for updating the operating parameters and BS selections, respectively. For a centralized traffic load balancing scheme, the network collects the operating status information from both BSs and users, and determines the user association for individual users.
A distributed traffic load balancing scheme can be implemented in a centralized fashion by leveraging virtualization techniques \cite{Han:2014:vGALA}. Thus, we propose a distributed traffic offloading procedure that includes four phases. The first phase is the initial user association and network utility measurements. When entering the network, a user simply attaches to any BS to retrieve network utility information. According to the initial user association, BSs measure their traffic load and estimate the traffic load in their backhaul. Based on the measurements, BSs, in the second phase, advertise their network utility information. Denote $\psi(\boldsymbol{\eta})=\sum_{j\in\mathcal{B}}w_{j}(\rho_{j})(\mu_{j}(\rho_{j})+\tilde{\mu}_{j}(\tilde{\rho}_{j}))$. In the third phase, the users select their serving BSs according to the advertised network utility information and the downlink data rates to minimize $\psi(\boldsymbol{\eta})$. In the fourth phase, the BSs and users iteratively update their network utilities (the second phase) and BS selections (the third phase), respectively, until the user association converges.

\subsection{The network utility aware user association}
The network utility aware user association scheme consists of a user side algorithm and a BS side algorithm. The user side algorithm based on the network utility advertisement selects the optimal serving BS for individual users while the BS side algorithm updates individual BSs' network utility advertisements based on the user association. In designing the network utility aware user association scheme, we make the following assumptions:
\begin{enumerate}
\item We assume that the time scale of the traffic arrival and departure process is faster relative to that of BSs in advertising their network utility information. That is to say, BSs advertise their network utility information after the system exhibits the stationary performance.
\item We assume that the green power generation rate is consistent during the time period of establishing a stable user association \cite{Han:2013:GALA}.
\item We assume that all the BSs are synchronized and advertise their network utility simultaneously and the system parameter $\kappa$ does not change during one user association process.
\item We assume that a BS's cache hit ratio is constant within the duration of one user association.
\end{enumerate}

The feasible set for the traffic load balancing problem in Eq. (\ref{eq:object_network_goal}) is
\begin{align}
\label{eq:feasible_set}
\mathcal{F}=\lbrace &\boldsymbol{\eta}|0\leq\tilde{\rho}_{j}\leq 1-\epsilon,\nonumber\\
&0\leq\rho_{j}\leq 1-\epsilon,\; \sum_{j\in\mathcal{B}}\eta_{j}(x)=1,\nonumber \\
& \eta_{j}(x)=\{0,1\},\;\forall j\in\mathcal{B},\;\forall x \in\mathcal{A}\rbrace
\end{align}
Since $\eta_{j}(x)=\{0,1\}$, $\psi(\boldsymbol{\eta})$ is not continuous differentiable. In order to derive the user side algorithm and the BS side algorithm for the NUA traffic load balancing scheme, we relax the feasible set by letting $0\leq\eta_{j}(x)\leq 1$. Then, the relaxed feasible set is
\begin{align}
\label{eq:rlx_feasible_set}
\tilde{\mathcal{F}}=\lbrace &\boldsymbol{\eta}|0\leq\tilde{\rho}_{j}\leq 1-\epsilon,\nonumber\\
&0\leq\rho_{j}\leq 1-\epsilon,\; \sum_{j\in\mathcal{B}}\eta_{j}(x)=1,\nonumber \\
& 0\leq\eta_{j}(x)\leq 1,\;\forall j\in\mathcal{B},\;\forall x \in\mathcal{A}\rbrace
\end{align}
After presenting the NUA traffic load balancing scheme, we will prove that the proposed scheme achieves an optimal user association in the feasible set of the traffic load balancing problem.

We define the time interval between two consecutive network utility advertisements as a time slot. Let $\eta_{j}^{k}(x)$ denote whether a user at location $x$ associates with BS $j$ in the $k$th time slot. Denote $\rho_{j}(k)$ and $\tilde{\rho}_{j}(k)$ as the traffic load in BS $j$ and its backhaul in the $k$th time slot, respectively.
Let
\begin{align}
\label{eq:object_dev}
\phi_{j}(k)&=\frac{d\psi(\boldsymbol{\eta})}{d \eta_{j}(x)}\nonumber\\
&=\frac{\lambda(x)\nu(x)}{r_{j}(x)}e^{\kappa(\rho_{j}(k)-\hat{\rho}_{j})}(\frac{\kappa\rho_{j}(k)}{1-\rho_{j}(k)}+\frac{\kappa\tilde{\rho}_{j}(k)}{1-\tilde{\rho}_{j}(k)}
+\frac{1}{(1-\rho_{j}(k))^{2}})\nonumber\\
&+\lambda(x)\nu(x)e^{\kappa(\rho_{j}(k)-\hat{\rho}_{j})}\frac{1-\alpha_{j}}{R_{j}(1-\tilde{\rho}_{j}(x))^{2}}.
\end{align}
Define
\begin{equation}
\label{eq:bs_msg_1}
\theta^{a}_{j}(k)=e^{\kappa(\rho_{j}(k)-\hat{\rho}_{j})}(\frac{\kappa\rho_{j}(k)}{1-\rho_{j}(k)}+\frac{\kappa\tilde{\rho}_{j}(k)}{1-\tilde{\rho}_{j}(k)}+\frac{1}{(1-\rho_{j}(k))^{2}})
\end{equation}
and
\begin{equation}
\label{eq:bs_msg_2}
\theta^{b}_{j}(k)=e^{\kappa(\rho_{j}(k)-\hat{\rho}_{j})}\frac{1-\alpha_{j}}{R_{j}(1-\tilde{\rho}_{j}(x))^{2}}.
\end{equation}
Since $\theta^{a}_{j}(k)$ and $\theta^{b}_{j}(k)$ are calculated based on BS $j$'s network utility, we define $\theta^{a}_{j}(k)$ and $\theta^{b}_{j}(k)$ as the network utility information advertised by BS $j$ in the $k$th time slot.

\subsubsection{The User Side Algorithm}
At the beginning of the $k$th time slot, BSs broadcast their network utility advertisements, e.g., $\theta^{a}_{j}(k)$ and $\theta^{b}_{j}(k)$, to users. The BS selection rule for a user at location~$x$ is
\begin{equation}
\label{eq:bs_selection}
b^{k}(x)= \arg\max_{j \in \mathcal{B}}\frac{r_{j}(x)}{\theta^{a}_{j}(k)+r_{j}(x)\theta^{b}_{j}(k)}.
\end{equation}
Here, $b^{k}(x)$ is the index of the BS selected by the user. Therefore,
\begin{equation}
\label{eq:eta_ua_update}
\setlength{\nulldelimiterspace}{0pt}
\eta^{k}_{j}(x)=\left\{\begin{IEEEeqnarraybox}[\relax][c]{ls}
1,
\; &for $j=b^{k}(x),\forall x \in \mathcal{A}$  \\
0, \; &for $j\neq b^{k}(x),\forall x \in \mathcal{A}$,%
\end{IEEEeqnarraybox}\right.
\end{equation}
\subsubsection{The BS Side Algorithm}
After mobile users select their associating BSs, the user association in BS~$j$, $\boldsymbol{\eta}^{k}_{j}$, is updated. Give the user association, BS $j$ updates its network utility advertisement. BS~$j$ calculates an intermediate user association $\bar{\boldsymbol{\eta}}^{k}_{j}=\{\bar{\eta}^{k}_{j}|j\in\mathcal{B}\}$ as
\begin{equation}
\bar{\boldsymbol{\eta}}^{k}_{j}= (1-\delta^{k})\boldsymbol{\eta}^{k}_{j}+\delta^{k}\bar{ \boldsymbol{\eta}}^{k-1}_{j}.
\end{equation}
Here, $0<\delta^{k}<1$ is an exponential averaging parameter. With the intermediate user association, BS $j$ calculates the intermediate traffic load in the BS and its backhaul. The intermediate traffic load in BS~$j$'s backhaul is
\begin{equation}
\label{eq:backhaul_load_ad}
\tilde{\rho}_{j}(k+1)=\int_{x \in \mathcal{A}}\frac{\lambda(x)\nu(x)\bar{\eta}^{k}_{j}(x)}{R_{j}}dx,
\end{equation}
and intermediate traffic load in BS~$j$ is
\begin{equation}
\label{eq:bs_load_ad}
\rho_{j}(k+1)=\int_{x \in \mathcal{A}}\frac{\lambda(x)\nu(x)\bar{\eta}^{k}_{j}(x)}{r_{j}(x)}dx.
\end{equation}
Based on the intermediate traffic load in both the BS and its backhaul, BS $j$ calculates its network utility advertisements, $\theta^{a}_{j}(k+1)$ and $\theta^{b}_{j}(k+1)$, using Eqs. (\ref{eq:bs_msg_1}) and  (\ref{eq:bs_msg_2}).
\subsection{The properties of the network utility aware user association}
In this subsection, we prove the convergence and the optimality of the proposed NUA traffic load balancing scheme. Since users select serving BSs based on BSs' network utility advertisements, the user association converges when BSs' network utility advertisements are stablilized. A BS's network utility advertisements are determined by its intermediate traffic loads in the BS and its backhaul. On calculating the intermediate traffic loads, the intermediate user association is the only variable. Therefore, when the intermediate user association converges, the intermediate traffic load is stabilized and so are the network utility advertisements.
Therefore, we first prove any BS's network utility advertisements are stabilized by proving that its intermediate user association converges and then show that the user association based on the stabilized network utility advertisements minimizes $\psi(\boldsymbol{\eta})$.


\begin{lemma}
\label{thm:feasible_set}
The relaxed feasible set $\mathcal{\tilde{F}}$ is a convex set.
\end{lemma}
\begin{proof}
The lemma is proved by showing that the set $\mathcal{\tilde{F}}$ contains any convex combination of the user association vector~$\boldsymbol{\eta}$.
\end{proof}

\begin{lemma}
\label{thm:convexity}
$\psi(\boldsymbol{\eta})$ is a convex function of $\boldsymbol{\eta}$ when $\boldsymbol{\eta}$ is defined in $\mathcal{\tilde{F}}$.
\end{lemma}
\begin{proof}
\label{prf:cvx}
The lemma can be proved by showing $\bigtriangledown^{2}\psi(\boldsymbol{\eta})>0$ when $\boldsymbol{\eta}$ is defined in $\mathcal{\tilde{F}}$.
\end{proof}

Let $\bar{\boldsymbol{\eta}}^{k}=\{\bar{\boldsymbol{\eta}}^{k}_{j}|j\in\mathcal{B}\}$ and $\triangle\bar{\boldsymbol{\eta}}^{k}=\bar{\boldsymbol{\eta}}^{k}- \bar{\boldsymbol{\eta}}^{k-1}$.

\begin{lemma}
\label{thm:descent_direction}
When $\triangle\bar{\boldsymbol{\eta}}^{k} \neq 0$, $\triangle\bar{\boldsymbol{\eta}}^{k}$ provides a descent direction of $\psi(\bar{\boldsymbol{\eta}})$ at $\bar{\boldsymbol{\eta}}^{k}$.
\end{lemma}
\begin{proof}
\label{prf:descent_direction}
Since $0\leq\bar{\eta}_{j}^{k}\leq 1,\forall k,\;\forall j\in\mathcal{B}$, $\bar{\boldsymbol{\eta}}$ is defined in $\tilde{\mathcal{F}}$. According Lemmas \ref{thm:feasible_set} and \ref{thm:convexity}, $\psi(\bar{\boldsymbol{\eta}})$ is a convex function of $\bar{\boldsymbol{\eta}}$. Hence, the lemma can be proved by showing $\langle\bigtriangledown{\psi(\bar{\boldsymbol{\eta}})}|_{\bar{\boldsymbol{\eta}}=\bar{\boldsymbol{\eta}}^{k}},\triangle\bar{\boldsymbol{\eta}}^{k}\rangle<0$.

\begin{eqnarray}
\label{eq:prf_descent_direction}
&&\langle\bigtriangledown{\psi(\bar{\boldsymbol{\eta}})}_{\bar{\boldsymbol{\eta}}=\bar{\boldsymbol{\eta}}^{k}},\triangle\bar{\boldsymbol{\eta}}^{k}\rangle\\
&&=\int_{x \in\mathcal{A}} \sum_{j\in\mathcal{B}}\lambda(x)\nu(x)(\bar{\eta}^{k}_{j}(x)-\bar{\eta}^{k-1}_{j}(x))(\frac{\theta^{a}_{j}(k)}{r_{j}(x)}+\theta^{b}_{j}(k))\nonumber\\
&&=(1-\delta^{k})\int_{x \in \mathcal{A}}\lambda(x)\nu(x)\sum_{j\in\mathcal{B}}(\eta^{k}_{j}(x)-\bar{\eta}^{k-1}_{j}(x))\frac{\theta^{a}_{j}(k)+r_{j}(x)\theta^{b}_{j}(k)}{r_{j}(x)} \nonumber\\
\end{eqnarray}
Since
\begin{equation}
\label{eq:eta_m}
\setlength{\nulldelimiterspace}{0pt}
\eta^{k}_{j}(x)=\left\{\begin{IEEEeqnarraybox}[\relax][c]{ls}
1,
\; &for $j=b^{k}(x)$  \\
0, \; &for $j\neq b^{k}(x)$,%
\end{IEEEeqnarraybox}\right.
\end{equation}

\begin{equation}
\label{eq:sum_leq_0}
\sum_{j\in\mathcal{B}}(\eta^{k}_{j}(x)-\bar{\eta}^{k-1}_{j}(x))\frac{\theta^{a}_{j}(k)+r_{j}(x)\theta^{b}_{j}(k)}{r_{j}(x)} \leq 0.
\end{equation}
Because $0<\delta^{k}<1$ and $\triangle\bar{\boldsymbol{\eta}}_{j}^{k} \neq 0$,
\begin{equation}
\label{eq:sum_less_0}
\sum_{j\in\mathcal{B}}(\eta^{k}_{j}(x)-\bar{\eta}^{k-1}_{j}(x))\frac{\theta^{a}_{j}(k)+r_{j}(x)\theta^{b}_{j}(k)}{r_{j}(x)} < 0.
\end{equation}
Thus, $\langle\bigtriangledown{\psi(\bar{\boldsymbol{\eta}})}|_{\bar{\boldsymbol{\eta}}=\bar{\boldsymbol{\eta}}^{k}},\triangle\bar{\boldsymbol{\eta}}^{k}\rangle<0$.
\end{proof}
Denote $\bar{\boldsymbol{\eta}}^{*}$ as the optimal intermediate user association.

\begin{lemma}
\label{thm:exit_delta}
When $\bar{\boldsymbol{\eta}}^{k}\neq\bar{\boldsymbol{\eta}}^{*}, \;\bar{\boldsymbol{\eta}}^{k}\in\tilde{\mathcal{F}}$, there exists $0<\delta^{k}<1$ such that $\psi(\bar{\boldsymbol{\eta}}^{k})< \psi(\bar{\boldsymbol{\eta}}^{k-1})$.
\end{lemma}

\begin{proof}
\label{prf:exit_delta}
Since
\begin{eqnarray}
\label{eq:exit_delta}
\triangle\bar{\boldsymbol{\eta}}^{k}&&=\bar{\boldsymbol{\eta}}^{k}- \bar{\boldsymbol{\eta}}^{k-1}\\
&&=(1-\delta^{k})(\boldsymbol{\eta}^{k}- \bar{\boldsymbol{\eta}}^{k-1}), \nonumber
\end{eqnarray}
$(\boldsymbol{\eta}^{k}- \bar{\boldsymbol{\eta}}^{k-1})$, according to Lemma \ref{thm:descent_direction}, provides the descent direction for searching the optimal value in the iterations while $(1-\delta^{k})$ indicates the search step in the $k$th iteration. Since $\bar{\boldsymbol{\eta}}^{k}\neq\bar{\boldsymbol{\eta}}^{*}$, there exists $0<\delta^{k}<1$ that enables $\psi(\bar{\boldsymbol{\eta}}^{k})< \psi(\bar{\boldsymbol{\eta}}^{k-1})$
\end{proof}

\begin{theorem}
\label{thm:alg_converge}
If the traffic load balancing problem is feasible\footnote{The problem is feasible when the feasible set of the problem is not empty.} and $\delta^{k}$ is properly selected, $\bar{\boldsymbol{\eta}}^{k}=(1-\delta^{k})\boldsymbol{\eta}^{k}+\delta^{k}\bar{\boldsymbol{\eta}}^{k-1}$ converges to $\bar{\boldsymbol{\eta}}^{*}$.
\end{theorem}
\begin{proof}
\label{prf:alg_converge}
Since 1) $\bar{\boldsymbol{\eta}}^{k}_{j}-\bar{\boldsymbol{\eta}}^{k-1}_{j}$ is a descent direction of $\psi(\bar{\boldsymbol{\eta}})$ at $\bar{\boldsymbol{\eta}}^{k}$ and 2) $\delta^{k}$ is properly selected such that $\psi(\bar{\boldsymbol{\eta}}^{k})< \psi(\bar{\boldsymbol{\eta}}^{k-1})$, the mapping, $\bar{\boldsymbol{\eta}}^{k}=(1-\delta^{k})\boldsymbol{\eta}^{k}+\delta^{k}\bar{\boldsymbol{\eta}}^{k-1}$, keep decreasing $\psi(\bar{\boldsymbol{\eta}})$. Since $\psi(\bar{\boldsymbol{\eta}})\geq 0$, $\bar{\boldsymbol{\eta}}^{k}$ will eventually converge. According to Lemma \ref{thm:exit_delta}, $\bar{\boldsymbol{\eta}}^{k}$ converges to $\bar{\boldsymbol{\eta}}^{*}$. Otherwise, $\psi(\bar{\boldsymbol{\eta}})$ can be further reduced.
\end{proof}

\begin{corollary}
\label{thm:ad_converge}
Any BS's network utility advertisements, $\theta^{a}_{j}(k)$ and $\theta^{a}_{j}(k)$, $j\in\mathcal{B}$, are stabilized.
\end{corollary}
\begin{proof}
\label{prf:ad_converge}
$\theta^{a}_{j}(k)$ and $\theta^{a}_{j}(k)$, $j\in\mathcal{B}$, are calculated by the traffic load in BS $j$ and its backhaul, respectively. When the intermediate user association converges, the traffic loads are determined. As a result, individual BS's network utility advertisements are stabilized.
\end{proof}


\begin{theorem}
\label{thm:alg_opt}
Given that the traffic load balancing problem is feasible, the user association, $\boldsymbol{\eta}^{*}_{j}=\{\eta^{*}_{j}(x)|\eta^{*}_{j}(x)=\{0,1\},\;x\in\mathcal{A}\}$, $j\in\mathcal{B}$, based on the stabilized network utility advertisements is an optimal solution to the traffic load balancing problem.
\end{theorem}
\begin{proof}
\label{prf:alg_opt}
Denote $\theta^{a}_{j}(*)$ and $\theta^{b}_{j}(*)$ as BS $j$'s stabilized network utility advertisements. Let $\boldsymbol{\eta}^{*}=\{\boldsymbol{\eta}^{*}_{j}|j\in\mathcal{B}\}$ and $\boldsymbol{\eta}=\{\boldsymbol{\eta}_{j}|j\in\mathcal{B}\}$.  Here, $\boldsymbol{\eta}_{j}=\{\eta_{j}(x)|\eta_{j}(x)=\{0,1\},\;x\in\mathcal{A}\}$. Suppose $\boldsymbol{\eta}$ to be arbitrary user association in the feasible set $\mathcal{F}$ that is not equal to $\boldsymbol{\eta}^{*}$.

\begin{eqnarray}
\label{eq:prf_opt}
&&\langle\bigtriangledown{\psi(\boldsymbol{\eta})}|_{\boldsymbol{\eta} =\boldsymbol{\eta}^{*}},\boldsymbol{\eta}-\boldsymbol{\eta}^{*}\rangle\\
&&=\int_{x \in\mathcal{A}} \sum_{j\in\mathcal{B}}\lambda(x)\nu(x)(\eta_{j}(x)-\eta^{*}_{j}(x))(\frac{\theta^{a}_{j}(*)}{r_{j}(x)}+\theta^{b}_{j}(*))\nonumber\\
&&=\int_{x \in\mathcal{A}} \lambda(x)\nu(x)\sum_{j\in\mathcal{B}}(\eta_{j}(x)-\eta^{*}_{j}(x))(\frac{\theta^{a}_{j}(*)}{r_{j}(x)}+\theta^{b}_{j}(*))\nonumber
\end{eqnarray}
Since
\begin{equation}
\label{eq:bs_select_prf}
b^{*}(x)= \arg\max_{j \in \mathcal{B}}\frac{r_{j}(x)}{\theta^{a}_{j}(*)+r_{j}(x)\theta^{b}_{j}(*)}
\end{equation}
and
\begin{equation}
\label{eq:eta_m_opt}
\setlength{\nulldelimiterspace}{0pt}
\eta^{*}_{j}(x)=\left\{\begin{IEEEeqnarraybox}[\relax][c]{ls}
1,
\; &for $j=b^{*}(x)$  \\
0, \; &for $j\neq b^{*}(x)$,%
\end{IEEEeqnarraybox}\right.
\end{equation}
\begin{equation}
\label{eq:opt_comp}
\sum_{j\in\mathcal{B}}\eta_{j}(x)(\frac{\theta^{a}_{j}(*)}{r_{j}(x)}+\theta^{b}_{j}(*))
\geq \sum_{j\in\mathcal{B}}\eta^{*}_{j}(x)(\frac{\theta^{a}_{j}(*)}{r_{j}(x)}+\theta^{b}_{j}(*)).
\end{equation}
Hence, $\langle\bigtriangledown{\psi(\boldsymbol{\eta})}|_{\boldsymbol{\eta} =\boldsymbol{\eta}^{*}},\boldsymbol{\eta}-\boldsymbol{\eta}^{*}\rangle\geq 0$. Therefore, $\boldsymbol{\eta}^{*}$ is an optimal solution to the UA problem.
\end{proof}
\subsection{The adaptation of the energy-latency tradeoff}
The system parameter, $\kappa$, controls the tradeoff between the green power utilization and the traffic delivery latency. When $\kappa=0$, $w_{j}(\rho_{j})=1$. In this case, the green power utilization is not modeled in the objective function. Thus, the NUA traffic load balancing scheme determines the user association based only on the traffic delivery latency. As $\kappa$ increases, the awareness of green power utilization in determining the user association enhances. In other words, with a larger $\kappa$, the green power utilization plays a more important role in determining the user association. If $\kappa$ is large enough, the NUA traffic load balancing scheme achieves a user association that approximates the user association that only cares about the green power utilization.
\section{Simulation and Performance Evaluation}
\label{sec:simulation}
\subsection{Simulation setup}
\begin{table}[ht]
\caption{Channel Model and Parameters}
\centering
\begin{tabular}{l||l}
\hline
Parameters & Value\\
\hline
$PL_{MBS}$ (dB) & $PL_{MBS}=128.1+37.6\log_{10}(d)$\\
$PL_{SCBS}$ (dB) & $PL_{SCBS}=38+10\log_{10}(d)$ \\
Rayleigh fading & 9 $dB$\\
Shadowing fading & 5 $dB$ \\
Antenna gain & 15 $dB$\\
Noise power level & -174 $dBm$ \\
Receiver sensitivity & -123 $dBm$ \\
\hline
\end{tabular}
\label{table:sim_parameters}
\end{table}

We set up system level simulations to investigate the performance of the NUA traffic load balancing scheme for the downlink traffic load balancing in backhaul constrained SCNs. In the simulation, three MBSs and seven SCBSs are randomly deployed in a $2000 m \times 2000 m$ area. The total bandwidth is 10~$MHz$ and the frequency reuse factor is one. The channel propagation model is based on COST 231 Walfisch-Ikegami~\cite{COST231}. The channel model and parameters are summarized in Table~\ref{table:sim_parameters}. Here, $PL_{MBS}$ and $PL_{SCBS}$ are the path loss between the users and MBSs and SCBSs, respectively. $d$ is the distance between users and BSs. The transmit power of an MBS and an SCBS are 43~$dBm$ and 33~$dBm$, respectively.

\begin{table}[ht]
\caption{The Average Cache Hit Ratio}
\centering
\begin{tabular}{|c|c|c|c|c|c|c|c|c|c|}
\hline
MBS 1 & MBS 2 & MBS 3 & SCBS 4 & SCBS 5 & SCBS 6 &SCBS 7 &SCBS 8&SCBS 9&SCBS 10 \\
\hline
0.27&0.12&0.28&0.12&0.17&0.22&0.22&0.24&0.24&0.19\\
\hline
\end{tabular}
\label{table:hit_ratio}
\end{table}

The static power consumptions of an MBS and an SCBS are 750~$W$ and 37~$W$, respectively~\cite{Auer:2011:HME}. The load-power coefficients of the MBS and the SCBS are 500 and 4, respectively~\cite{Auer:2011:HME}. The solar cell power efficiency is $17.4\%$~\cite{HIT:Photo}. We assume that the weather condition is the standard condition which specifies a temperature of 25~$^{o}C$, an irradiance of 1000~$W/m^{2}$, and an air mass of 1.5 spectrum \cite{Riordan:1990:WAA}. Thus, the green power generation rate is 174~$W/m^{2}$. The solar panel sizes are randomly selected but ensure the green power generation capacity of MBSs from 750~$w$ to 1300~$w$ while that of SCBSs from 37~$w$ to 48~$w$. BSs' energy-latency coefficients are set to be the same. In the simulation, the average data rate of SCBSs' backhaul is $5$~Mbps. The hit ratio\footnote{The cache hit ratio is randomly selected from 0.1 to 0.3. For the analytical simplicity, we fix the hit ratio of BSs in the simulations.} of BSs' cache system is shown in Table~\ref{table:hit_ratio}.

\subsection{Traffic load balancing algorithms and network utility awareness}
\begin{table}[ht]
\caption{Network Utility Aware User Association Schemes}
\centering
\begin{tabular}{|l|c|c|c|c|}
\hline
UA Scheme & green power & BS Latency & Backhaul Latency & Cache \\
\hline
\hline
vGALA ($\kappa=0$)&&x&&\\
\hline
vGALA ($\kappa=6,\theta=1$)&x&&&\\
\hline
vGALA ($\kappa=4,\theta=0.5$)&x&x&&\\
\hline
NUA ($\kappa=0$, $\alpha_{j}=0$)&&x&x&\\
\hline
NUA ($\kappa=0$, $\alpha_{j}>0$)&&x&x&x\\
\hline
NUA ($\kappa=2$, $\alpha_{j}=0$)&x&x&x&\\
\hline
NUA ($\kappa=2$, $\alpha_{j}>0$)&x&x&x&x\\
\hline
DRB-NU&x&x&x&x\\
\hline
\end{tabular}
\label{table:alg_compare}
\end{table}
In the simulations, we investigate the performance of traffic load balancing schemes with different levels of network utility awareness. The network utilities considered in this paper are green power, the traffic delivery latency in BSs, the traffic delivery latency in backhaul, and the cache hit ratio. We implement three traffic load balancing schemes in the simulations. The first scheme is vGALA \cite{Han:2014:vGALA}. Adapting the parameters of vGALA ($\kappa$ and $\theta$), we realize the three traffic load balancing schemes: 1) BS latency aware, 2) green power aware, and 3) BS latency and green power aware. The second scheme is the NUA traffic load balancing scheme that simulates four traffic load balancing schemes: 1) BS latency and backhaul latency aware, 2) BS latency, backhaul latency and cache hit ratio aware, 3) BS latency, backhaul latency and green power aware, and 4) all network utilities aware. For the NUA scheme, when $\alpha_{j}=0,\;\forall j\in\mathcal{B}$, a BS estimates the traffic delivery latency in backhaul purely based on the traffic arrival rates in the BS. In fact, if the cache system is considered, the traffic load in backhaul should be less than that in the BS. Therefore, the NUA scheme with $\alpha_{j}=0,\;\forall j\in\mathcal{B}$, simulates the cache unaware traffic load balancing scheme.
The third scheme is the data rate bias (DRB) scheme \cite{Damn:2011:S3GPP}. In the implementation, we assume that BSs in the same tier have the data rate bias. MBSs are in the first tier while SCBSs are in the second tier. In the data rate bias scheme, a user selects the serving BS to maximize the biased data rate. The data rate bias of an MBS is set to one. We vary the data rate bias of an SCBS to investigate the performance of the scheme. We set $\psi(\boldsymbol{\eta})$ as the performance metric for selecting the optimal data rate bias. Thus, the implemented DRB scheme is aware of all network utilities and is referred to as DRB-NU (the data rate bias with network utility awareness).
The network utility awareness of the schemes with different settings are shown in Table \ref{table:alg_compare}.

\subsection{Simulation results}
\begin{figure*}
\centering
\hspace*{\fill}
       \begin{subfigure}[b]{0.3\textwidth}
      	    \includegraphics[scale=0.2]{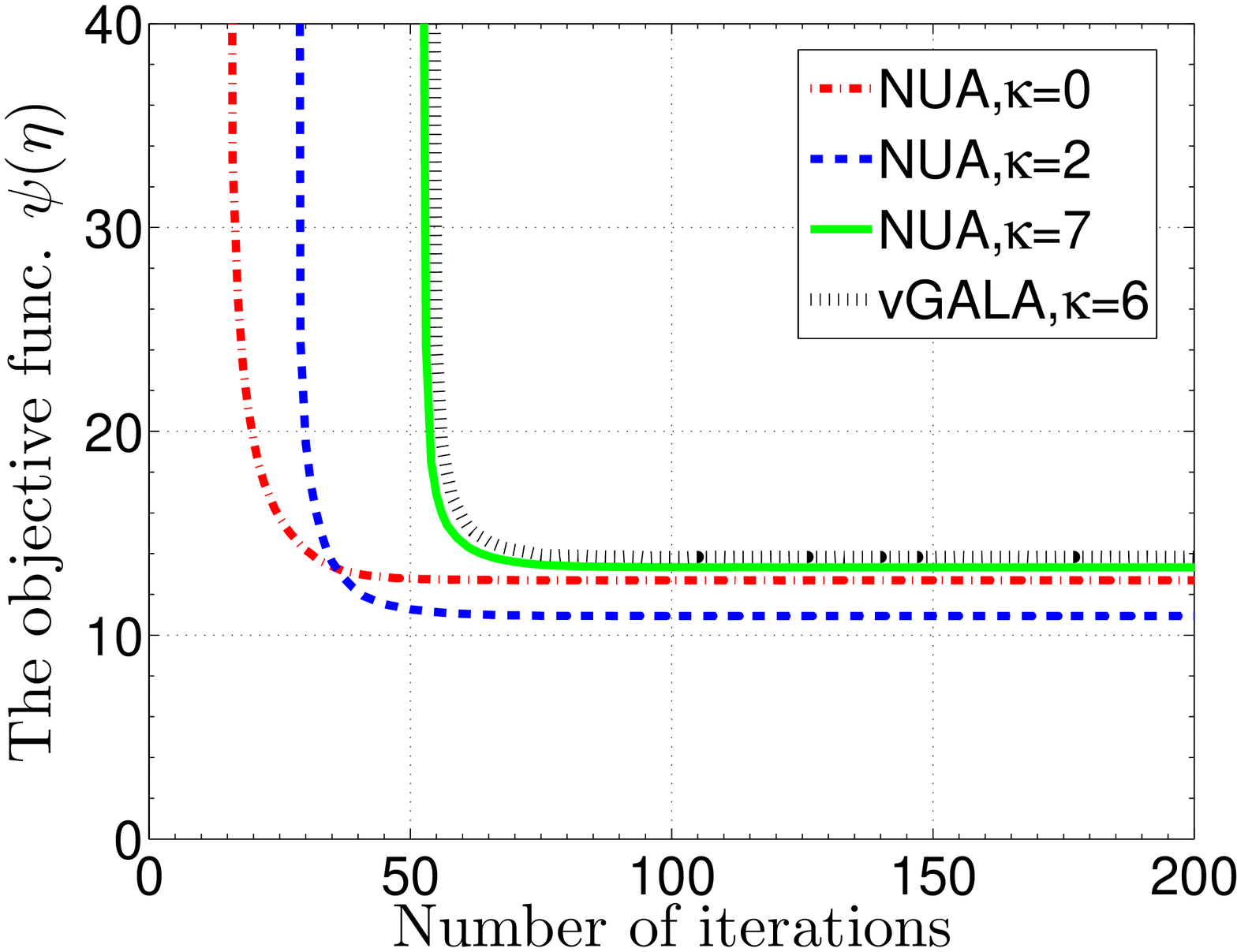}
            \caption{The value of $\psi(\boldsymbol{\eta})$.}
            \label{fig:sim_1_obj}
       \end{subfigure}\hfill
        \begin{subfigure}[b]{0.3\textwidth}
      	    \includegraphics[scale=0.2]{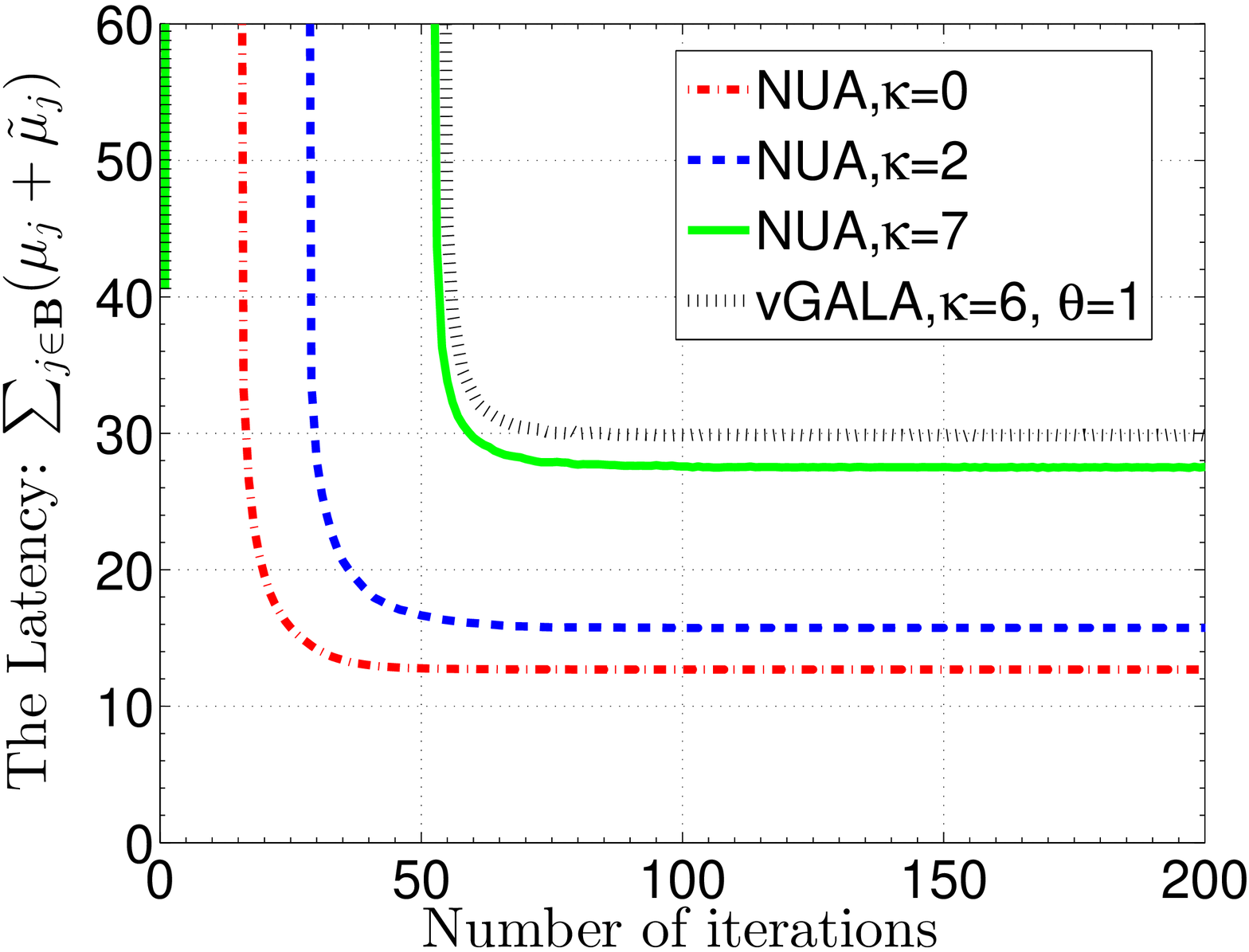}
            \caption{The traffic delivery latency.}
            \label{fig:sim_1_latency}
       \end{subfigure}\hfill
      \begin{subfigure}[b]{0.3\textwidth}
      	    \includegraphics[scale=0.2]{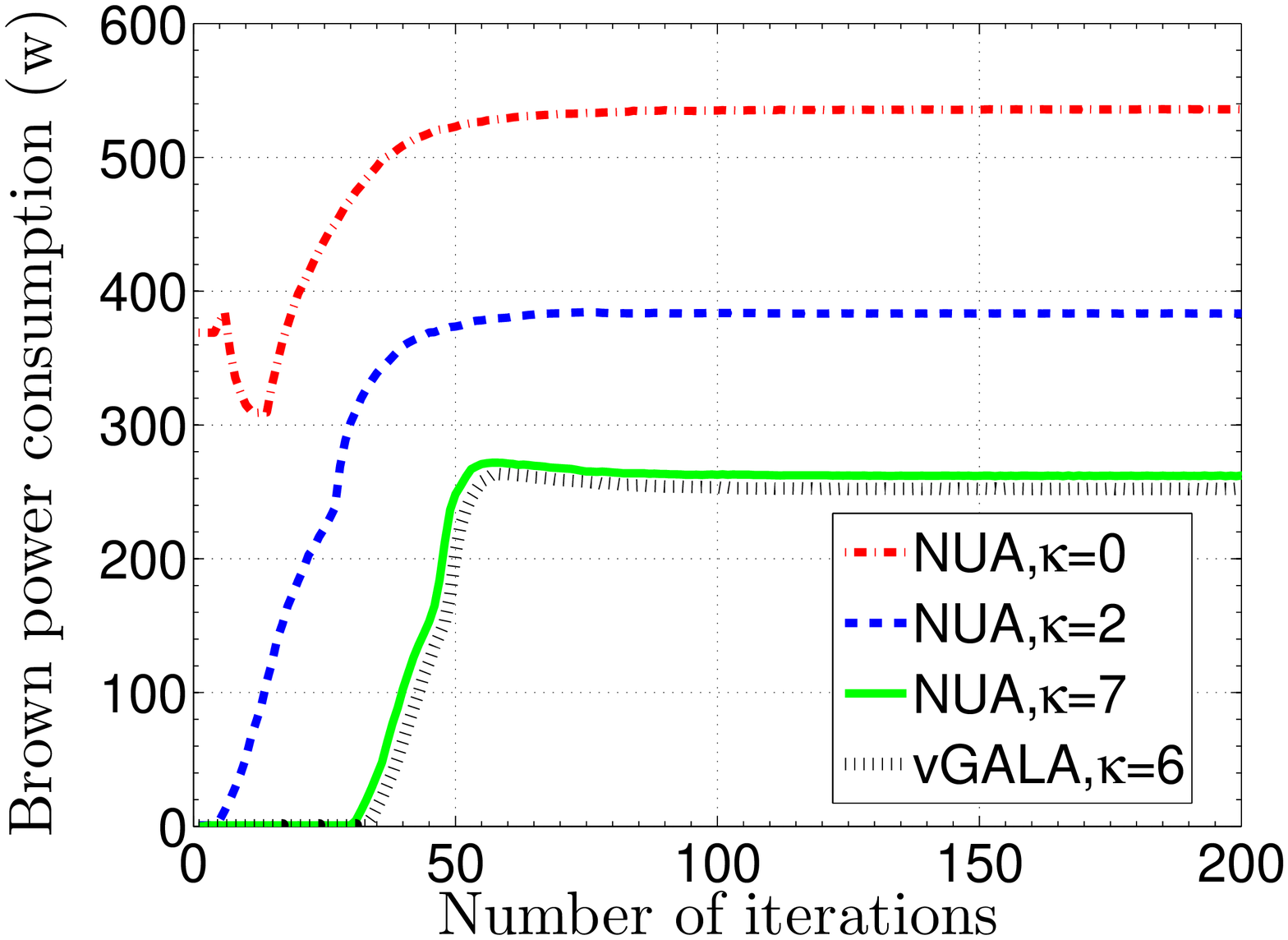}
            \caption{The brown power consumption.}
            \label{fig:sim_1_brown_energy}
       \end{subfigure}\hfill
    \caption{%
       The performance of the NUA scheme with different $\kappa$.
     }%
   \label{fig:sim_1_converge}
\end{figure*}
Fig.~\ref{fig:sim_1_converge} shows the convergence of the NUA scheme and its energy-latency tradeoff with different $\kappa$. Fig.~\ref{fig:sim_1_obj} shows the the value of $\psi(\boldsymbol{\eta})$ converges with less than 100 iterations, and so do the traffic delivery latency (Fig.~\ref{fig:sim_1_latency}) and the brown power consumption (Fig.~\ref{fig:sim_1_brown_energy}). Figs.~\ref{fig:sim_1_latency}~and~\ref{fig:sim_1_brown_energy} show the energy-latency tradeoff. As $\kappa$ increases, the network emphasizes the green power utilization in determining the user association. As a result, with a large $\kappa$, e.g., $\kappa=7$, the network consumes less brown power at the cost of introducing additional traffic delivery latency. The vGALA with a large $\kappa$ and $\theta$, e.g., $\kappa=6$ and $\theta=1$,  realizes the user association that is only aware of green power utilization \cite{Han:2014:vGALA}. Fig.~\ref{fig:sim_1_brown_energy} shows that as $\kappa$ increases, the performance of the NUA scheme in terms of the brown power consumption approaches that of the traffic load balancing scheme that optimizes the green power utilization (only aware of green power utilization).

Fig.~\ref{fig:sim_3_solar_rate} shows the performance of the NUA scheme versus different solar panel efficiency. Fig.~\ref{fig:sim_3_brown_power} shows that the brown power consumption reduces as the solar panel efficiency increases. This is because a higher solar panel efficiency enables solar panels to generate a larger amount of electricity and thus lessen the brown power consumption. As shown in Fig~\ref{fig:sim_3_latency}, the performance of the traffic delivery latency divides into four regions. In first region (R1), the traffic delivery latency does not change. This is because the green power generated in individual BSs is less than their static power consumption when the solar panel efficiency is in R1. In other words, the green capacity of all BSs is zero when the solar panel efficiency is within R1. As a result, increasing the solar panel efficiency in R1 does not impact the traffic delivery latency, and neither does the value of $\psi(\boldsymbol{\eta})$. In the second region, as shown in Fig.~\ref{fig:sim_3_latency}, the traffic delivery latency increases as the solar panel efficiency increases. When the solar panel efficiency is within this region, the network trades the traffic delivery latency for reducing the brown power consumption. In the third region (R3), the network trades the power consumption for reducing the traffic delivery latency. This can be seen from Fig.~\ref{fig:sim_3_brown_power}. The rate of brown power consumption reduction decreases when the solar panel efficiency is about 17\% (the start point of R3) as shown in Fig.~\ref{fig:sim_3_latency}. This indicates that the network emphasizes on reducing the traffic delivery latency when the solar panel efficiency is within R3. In both R2 and R3, the energy-latency tradeoff reduces the value of $\psi(\boldsymbol{\eta})$ as shown in Fig.~\ref{fig:sim_3_obj}. When the solar panel efficiency is within the fourth region (R4), the solar panel efficiency is high enough to enable zero brown power consumption in all BSs while minimizing the traffic delivery latency.
\begin{figure*}
\centering
\hspace*{\fill}
       \begin{subfigure}[b]{0.3\textwidth}
      	    \includegraphics[scale=0.2]{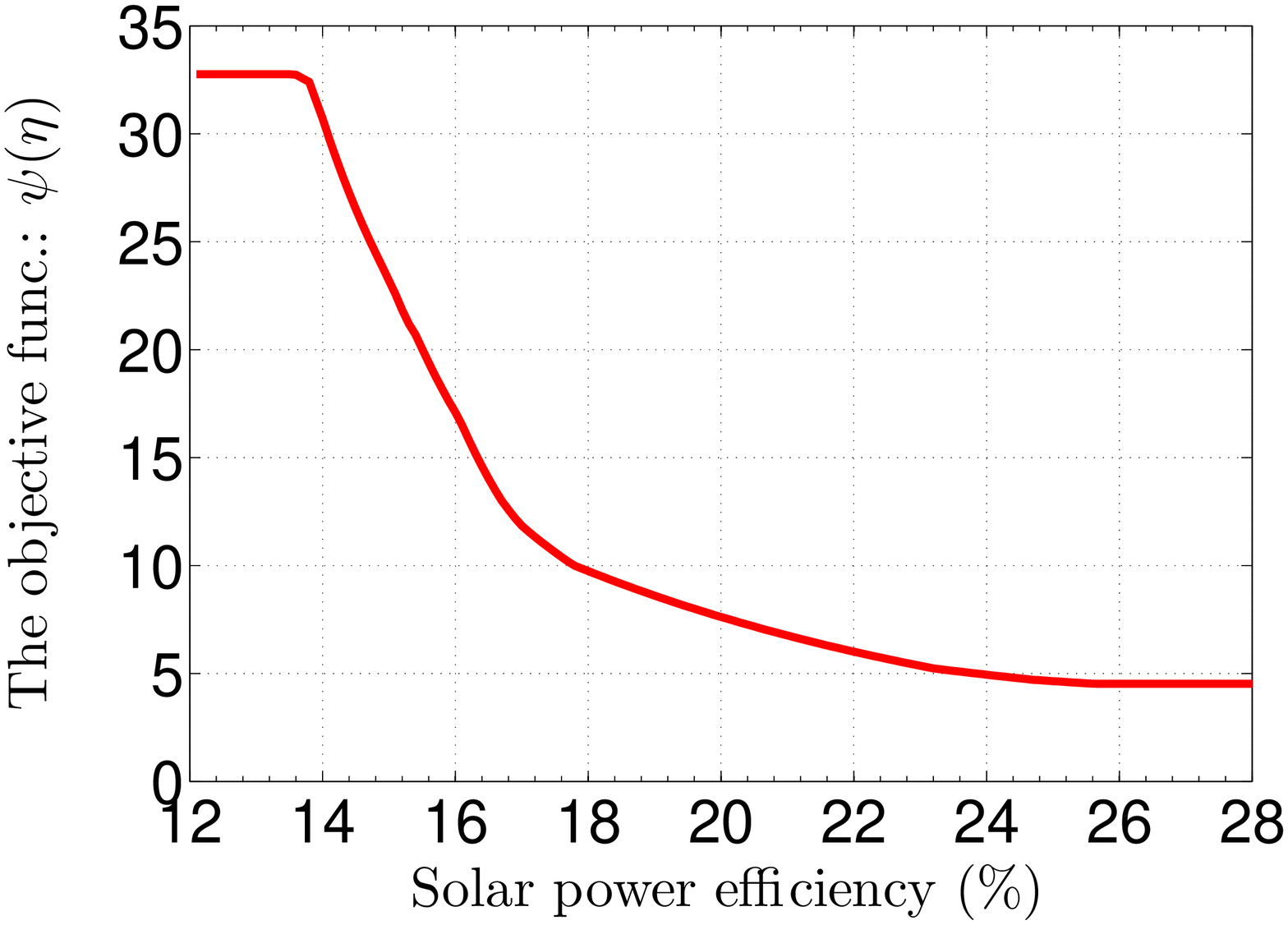}
            \caption{The value of $\psi(\boldsymbol{\eta})$.}
            \label{fig:sim_3_obj}
       \end{subfigure}\hfill
        \begin{subfigure}[b]{0.3\textwidth}
      	    \includegraphics[scale=0.2]{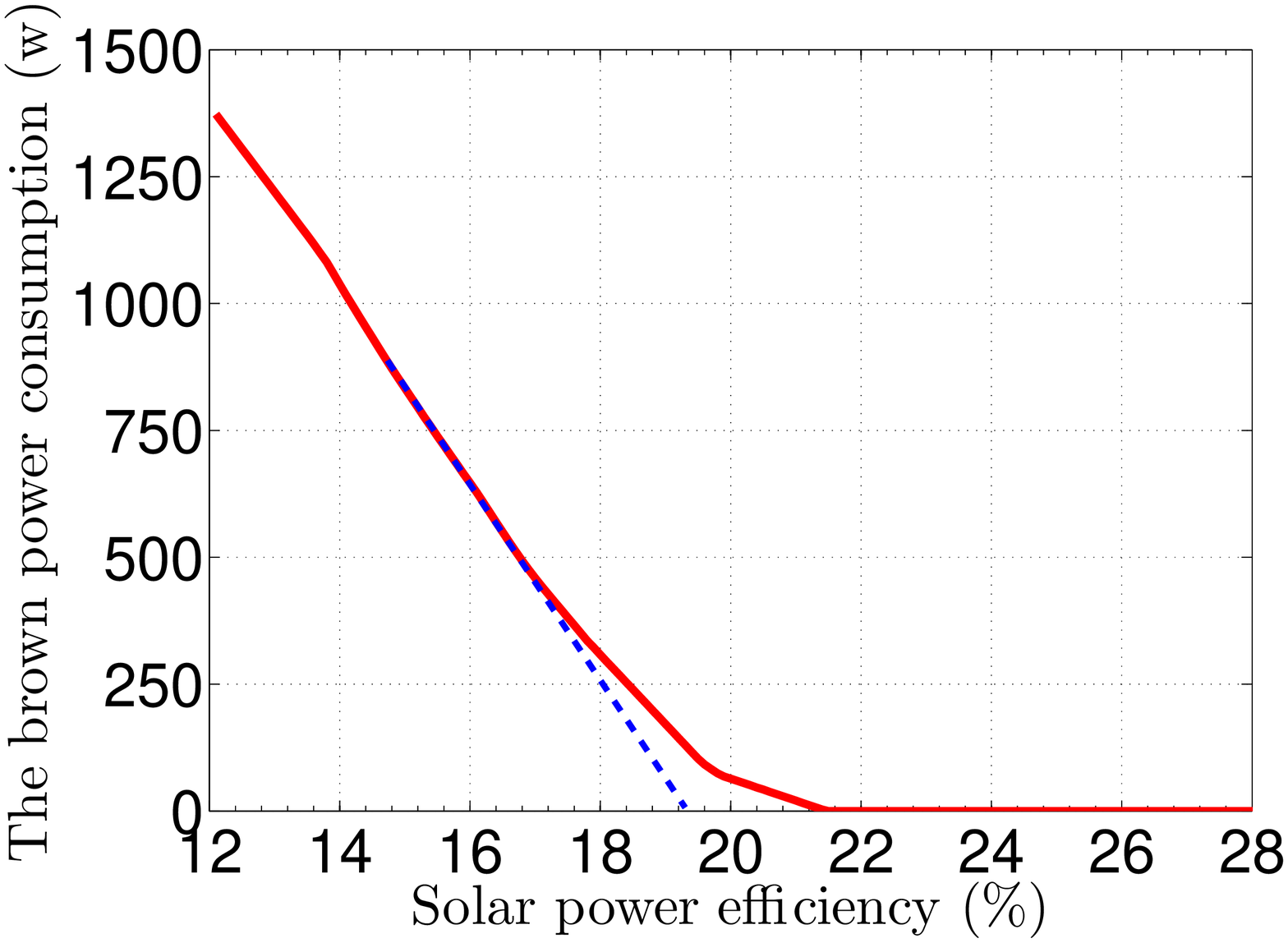}
            \caption{The brown power consumption.}
            \label{fig:sim_3_brown_power}
       \end{subfigure}\hfill
      \begin{subfigure}[b]{0.3\textwidth}
      	    \includegraphics[scale=0.2]{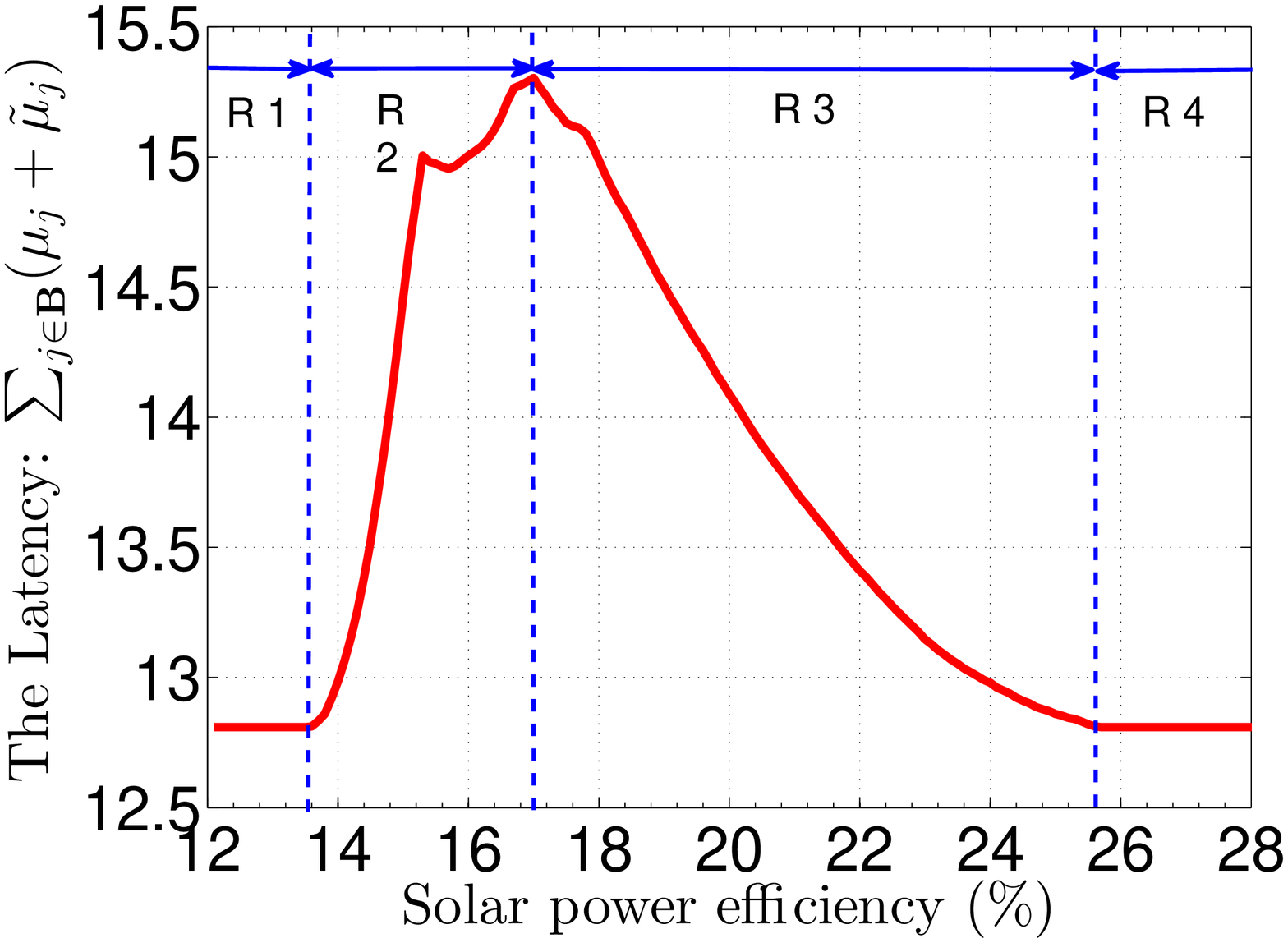}
            \caption{The traffic delivery latency.}
            \label{fig:sim_3_latency}
       \end{subfigure}\hfill
    \caption{%
       The performance of the NUA scheme versus the solar panel efficiency.
     }%
   \label{fig:sim_3_solar_rate}
\end{figure*}

\begin{figure}[!ht]
\centering
\includegraphics[scale=0.3]{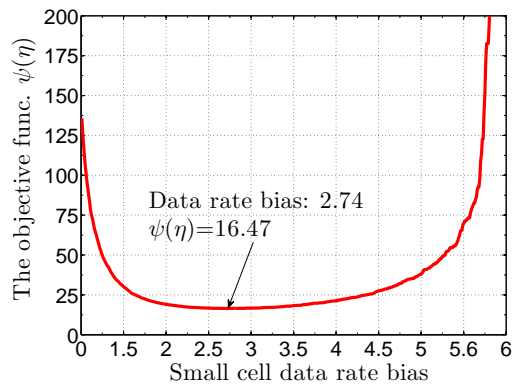}
\caption{The value of $\psi(\boldsymbol{\eta})$ versus data rate bias.}
\label{fig:sim_2_obj}
\end{figure}

Fig. \ref{fig:sim_2_obj} shows the value of $\psi(\boldsymbol{\eta})$ versus the small cell data rate biases under the DRB-NU scheme. The value of $\psi(\boldsymbol{\eta})$ is minimized when the small cell data rate bias equals to 2.74. Given the data rate bias, the network's traffic delivery latency is 17.24 and the brown power consumption is 480.9~$w$. Under the NUA scheme with $\kappa=2$, the network's traffic delivery latency and the brown power consumption are 15.74 and 383.35~$w$, respectively. Therefore, as compared with the DRB-NU scheme, the NUA scheme reduces the traffic delivery latency and the brown power consumption by 8.7\% and 20.28\%, respectively. The NUA scheme has achieved enhanced performance because it allows individual BSs to adapt their network utility advertisements while the DRB-NU scheme only allows to change the data rate bias for an entire tier rather rather than for individual BSs. Another drawback of the DRB-NU scheme is that it does not dynamically respond to the network utility changes. The small cell data rate bias is optimized based on the previous instead of current network conditions, e.g., traffic intensities, backhaul constraints, and green power availabilities.

In Figs. \ref{fig:sim_4_vary_bh}, \ref{fig:sim_6_limited_bh}, \ref{fig:sim_4_area_comp}, and \ref{fig:sim_6_area_comp}, we compare the performance of three traffic load balancing schemes with varying network conditions. The first scheme is the green power and BS latency aware traffic load balancing scheme realized by vGALA with $\kappa=4$ and $\theta=0.5$ \cite{Han:2014:vGALA}. The second one is the NUA scheme that is aware of all network utilities. The third one, referred to as NUA-NC (no cache), is the NUA without awareness of the cache hit ratio. This scheme is realized by the NUA scheme with $\alpha_{j}=0,\;\forall j\in\mathcal{B}$. Fig.~\ref{fig:sim_4_vary_bh} shows the performance of these traffic load balancing schemes versus different backhaul data rates. When the backhaul data rate is very low, e.g., less than 5 Mbps in the simulation, the value of $\psi(\boldsymbol{\eta})$ and the traffic delivery latency under vGALA is very large. This indicates that, without the awareness of backhaul data rates, the traffic load balancing under vGALA congests some BSs in the network. The brown power consumption under vGALA does not change versus the backhaul data rates because vGALA does not consider the traffic delivery latency in backhaul as a performance metric in determining the user association. As the backhaul data rates increase, the value of $\psi(\boldsymbol{\eta})$ under these schemes converges because the backhaul constraint is gradually mitigated. However, the NUA scheme achieves smaller traffic delivery latency as compared with the vGALA scheme because of the awareness of the traffic delivery latency in backhaul.

In the simulation, the value of $\psi(\boldsymbol{\eta})$ is minimized by the NUA scheme. However, the NUA-NC scheme achieves the minimal traffic delivery latency as shown in Fig.~\ref{fig:sim_4_latency}. This is because the NUA scheme aims to minimize the value of $\psi(\boldsymbol{\eta})$ and thus strikes a tradeoff between the traffic delivery latency and the brown power consumption. As a result, compared with the NUA-NC scheme, the NUA scheme consumes less brown power at the cost of an increase of the traffic delivery latency.

The BSs' coverage areas under these schemes are shown in Fig.~\ref{fig:sim_4_area_comp}. The NUA-NC scheme is aware of the backhaul limitation and thus reduces the coverage area of the backhaul constrained SCBS, e.g., SCBS~8. However, owing to the unawareness of the cache hit ratio, the NUA-NC scheme overestimates the traffic load in backhaul and constrains the coverage area of SCBSs, e.g., SCBS~8. The NUA scheme, being aware of the cache hit ratio of individual BSs, accurately estimates the traffic load in the backhaul and thus derives optimal coverage areas for BSs, e.g., increasing the coverage area of SCBS~8, to minimize the value of $\psi(\boldsymbol{\eta})$.

\begin{figure*}
\centering
\hspace*{\fill}
       \begin{subfigure}[b]{0.3\textwidth}
      	    \includegraphics[scale=0.2]{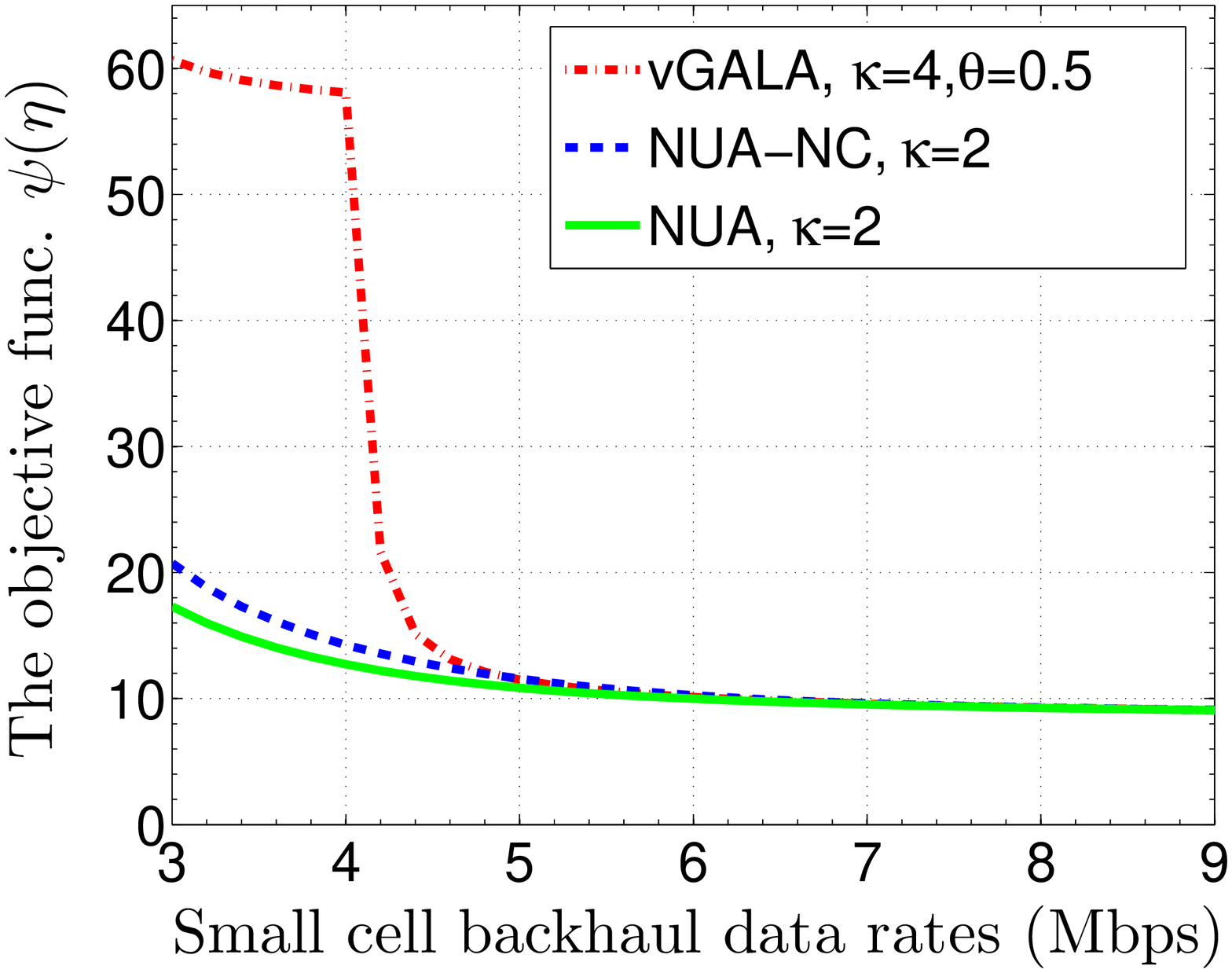}
            \caption{The value of $\psi(\boldsymbol{\eta})$.}
            \label{fig:sim_4_obj}
       \end{subfigure}\hfill
        \begin{subfigure}[b]{0.3\textwidth}
      	    \includegraphics[scale=0.2]{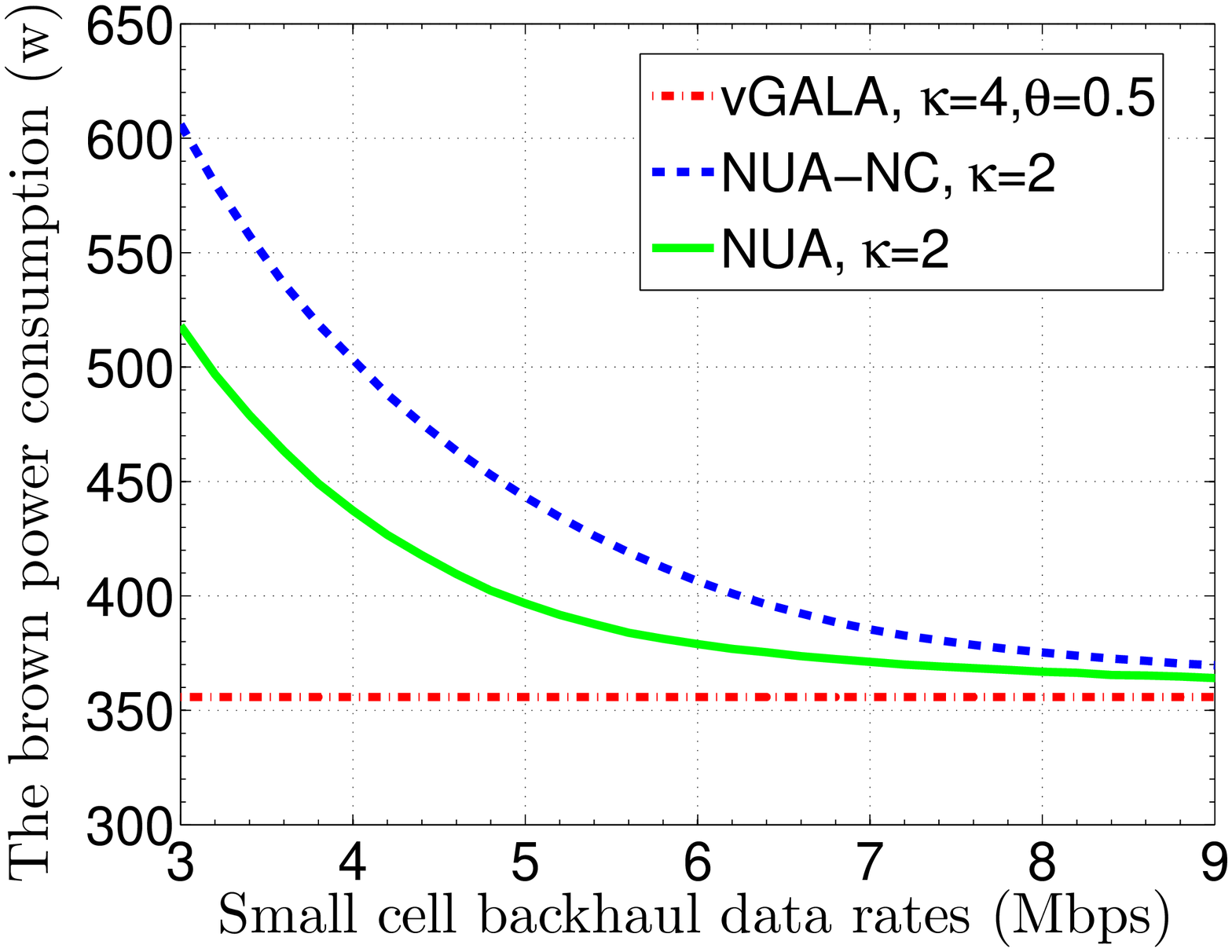}
            \caption{The brown power consumption.}
            \label{fig:sim_4_brown_power}
       \end{subfigure}\hfill
      \begin{subfigure}[b]{0.3\textwidth}
      	    \includegraphics[scale=0.2]{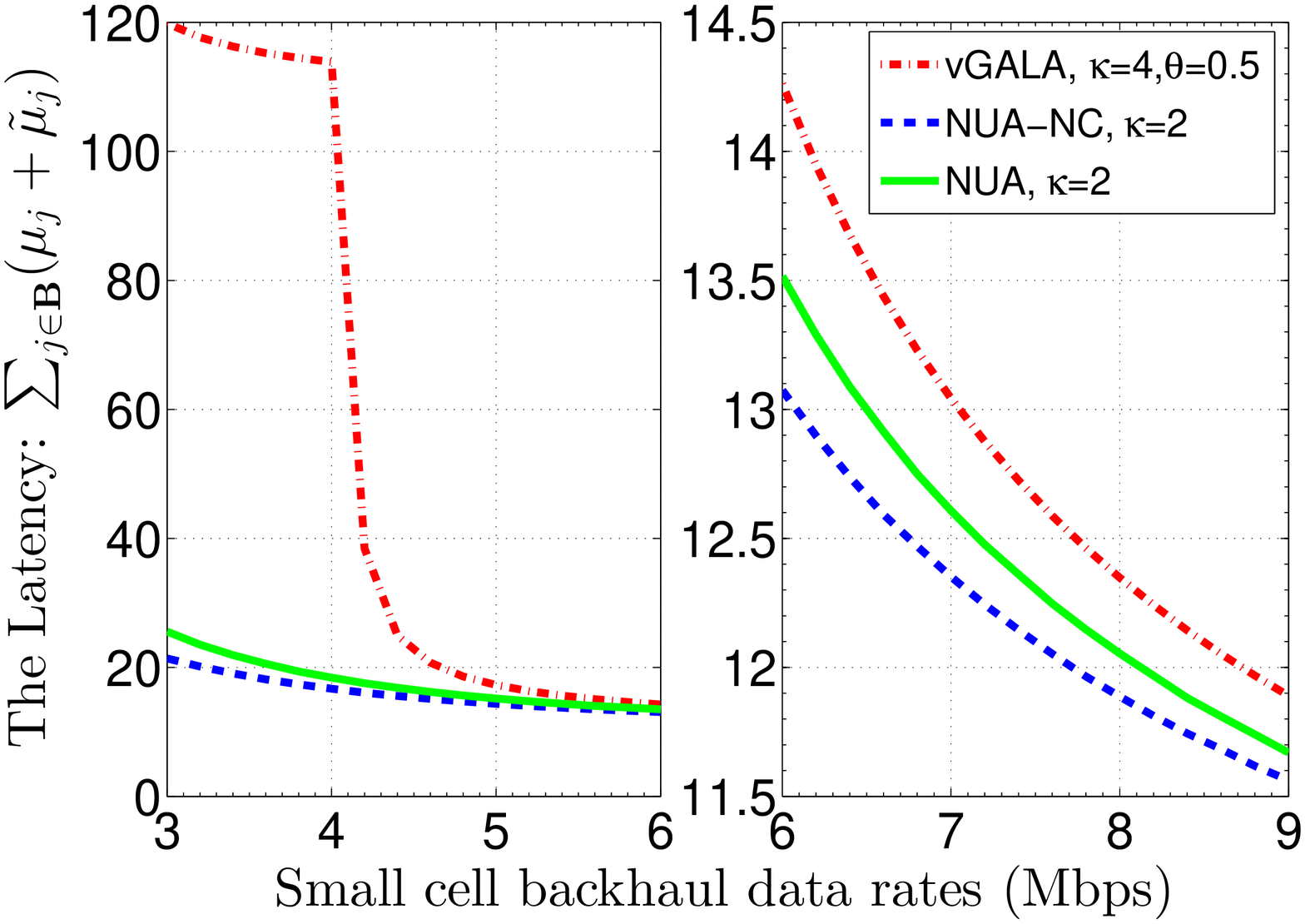}
            \caption{The traffic delivery latency.}
            \label{fig:sim_4_latency}
       \end{subfigure}\hfill
    \caption{%
       The performance of the traffic load balancing schemes versus the backhaul data rates.
     }%
   \label{fig:sim_4_vary_bh}
\end{figure*}
\begin{figure*}
\centering
\hspace*{\fill}
       \begin{subfigure}[b]{0.3\textwidth}
      	    \includegraphics[scale=0.2]{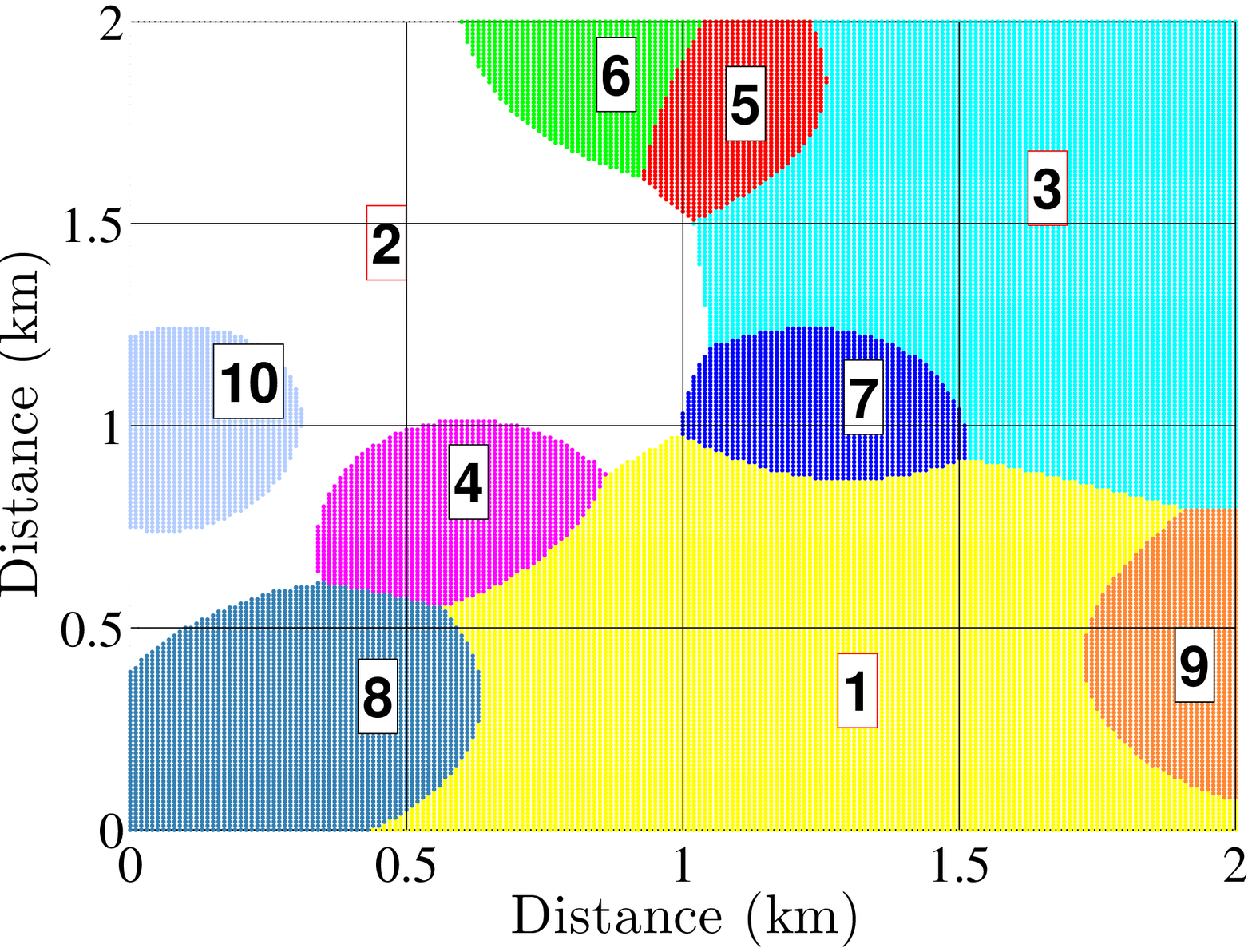}
            \caption{vGALA.}
            \label{fig:sim_4_vgala_area}
       \end{subfigure}\hfill
        \begin{subfigure}[b]{0.3\textwidth}
      	    \includegraphics[scale=0.2]{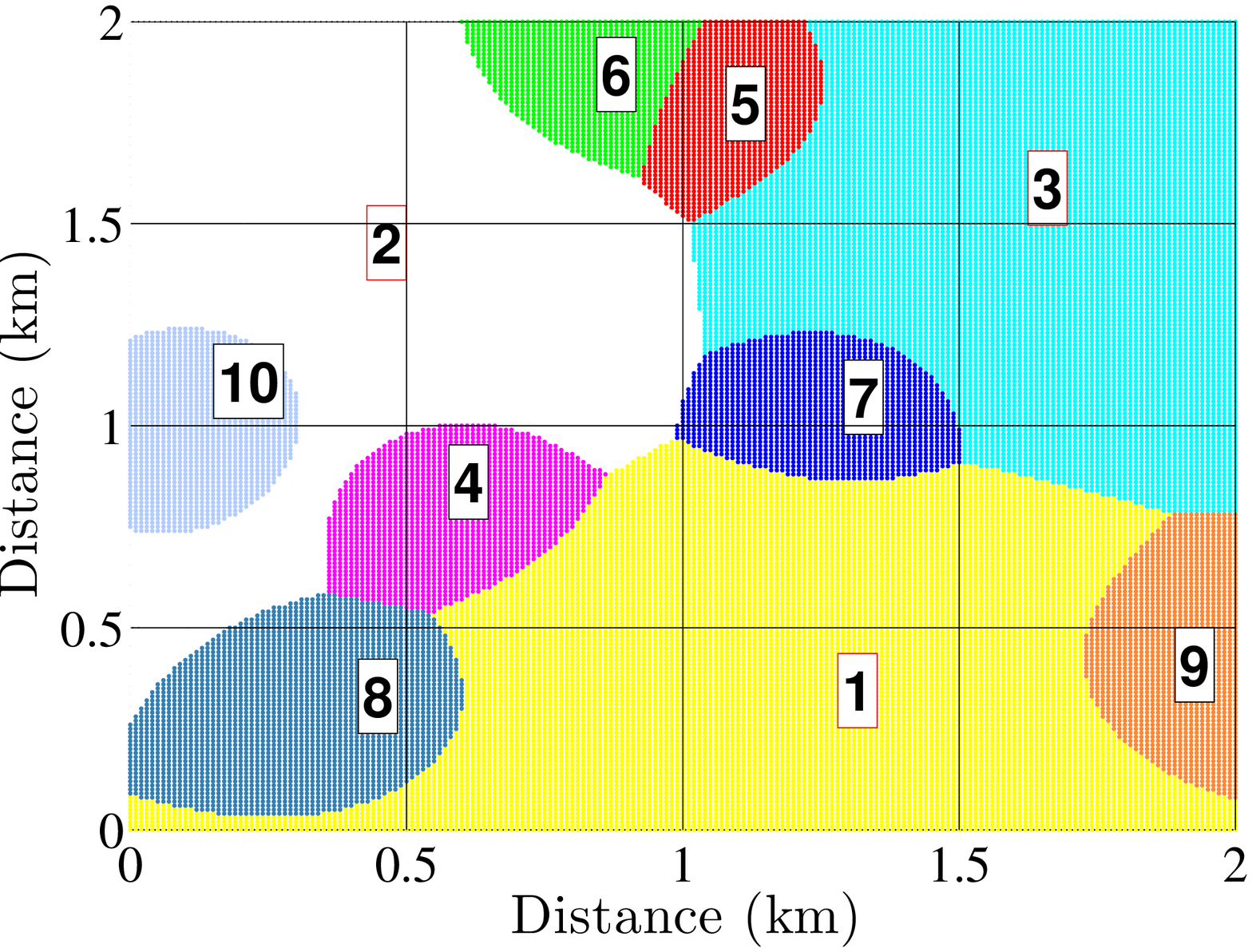}
            \caption{NUA-NC.}
            \label{fig:sim_4_nua_no_cache_area}
       \end{subfigure}\hfill
       \begin{subfigure}[b]{0.3\textwidth}
      	    \includegraphics[scale=0.2]{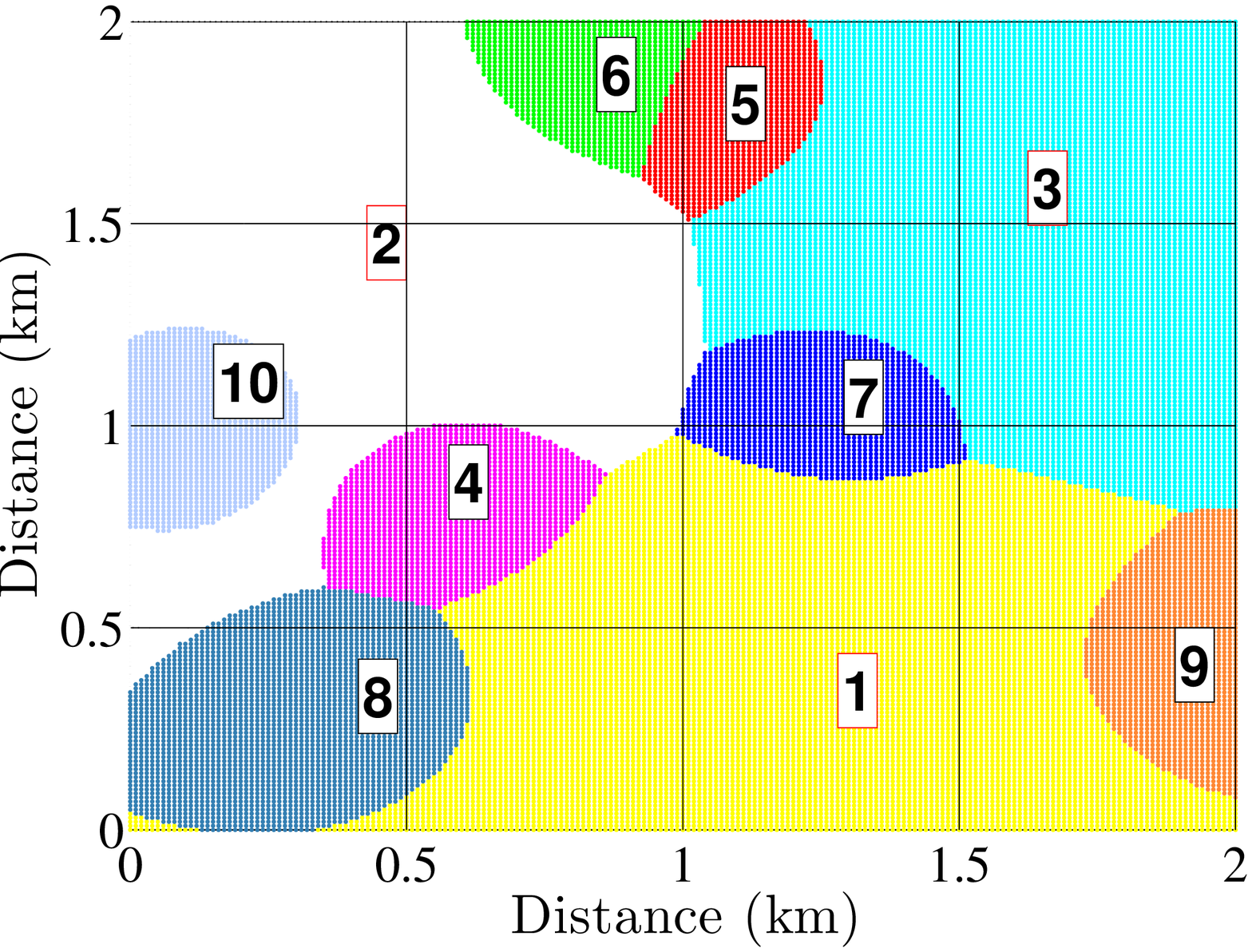}
            \caption{NUA.}
            \label{fig:sim_4_nua_area}
       \end{subfigure}\hfill
    \caption{%
       The coverage areas of different schemes ($R_{4-10}=5$ $Mbps$).
     }%
   \label{fig:sim_4_area_comp}
\end{figure*}

As shown in Fig. \ref{fig:sim_6_limited_bh}, when the backhaul data rate of a SCBS changes, e.g., $R_{5}$ reduces from 5 $Mbps$ to 1 $Mbps$, the NUA and NUA-NC scheme are able to adapt the traffic load balancing according to the backhaul data rate changes. However, the vGALA scheme, without the awareness of backhaul data rate, incurs excessive traffic delivery latency which is 667\% of the traffic delivery latency of the NUA scheme as shown in Fig.\ref{fig:sim_6_latency}. As shown in Fig. \ref{fig:sim_6_area_comp}, both the NUA and NUA-NC schemes are able to reduce the coverage area of SCBS 5. The NUA-NC scheme, because of the unawareness of the cache hit ratio, shrinks the coverage area more than the NUA scheme does.
\begin{figure*}
\centering
\hspace*{\fill}
       \begin{subfigure}[b]{0.3\textwidth}
      	    \includegraphics[scale=0.2]{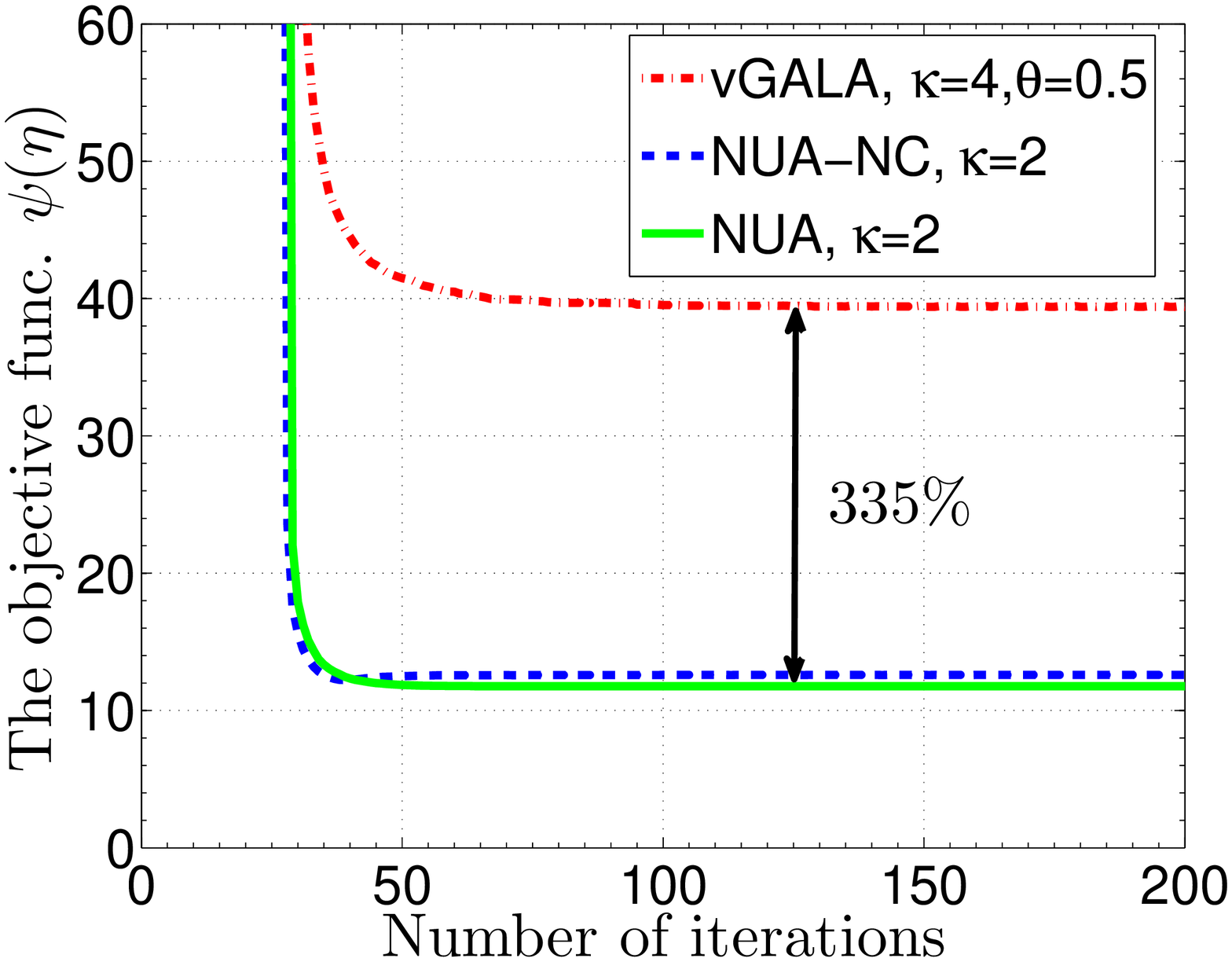}
            \caption{The value of $\psi(\boldsymbol{\eta})$.}
            \label{fig:sim_6_obj}
       \end{subfigure}\hfill
        \begin{subfigure}[b]{0.3\textwidth}
      	    \includegraphics[scale=0.2]{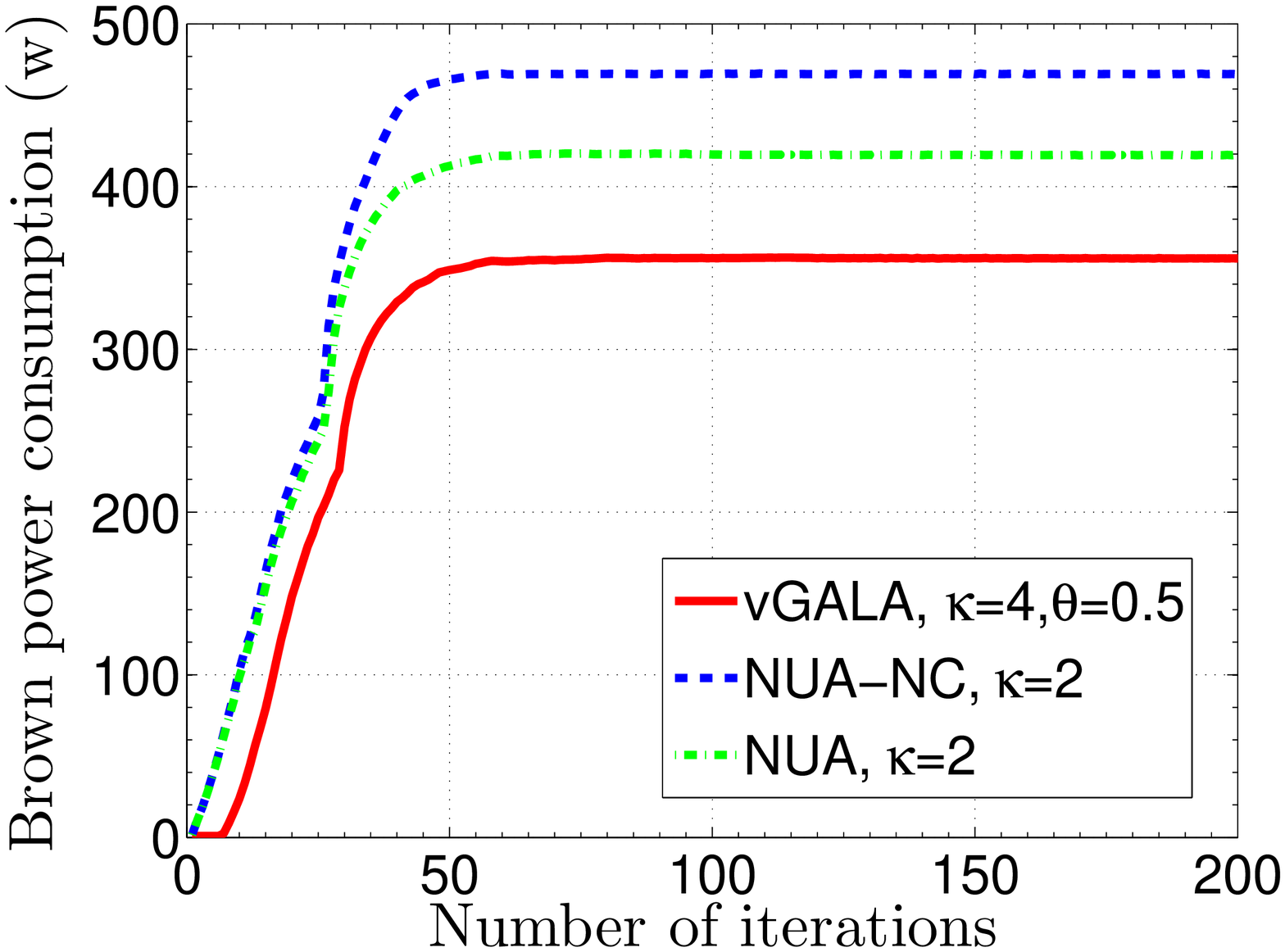}
            \caption{The brown power consumption.}
            \label{fig:sim_6_brown_power}
       \end{subfigure}\hfill
      \begin{subfigure}[b]{0.3\textwidth}
      	    \includegraphics[scale=0.2]{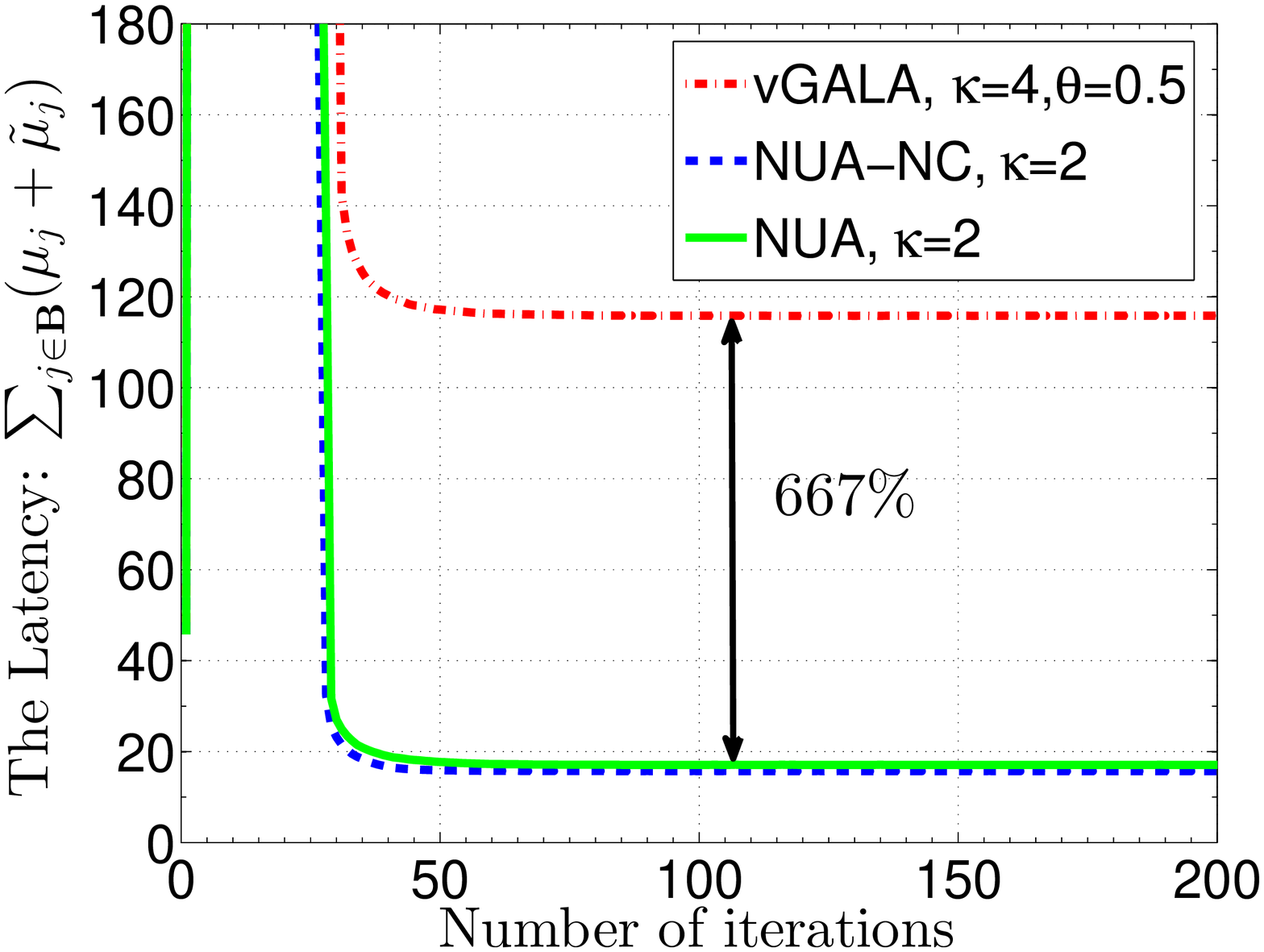}
            \caption{The traffic delivery latency.}
            \label{fig:sim_6_latency}
       \end{subfigure}\hfill
    \caption{%
       The performance comparison ($R_{5}=1$ $Mbps$ and $R_{4,6-10}=5$ $Mbps$).
     }%
   \label{fig:sim_6_limited_bh}
\end{figure*}

\begin{figure*}
\centering
\hspace*{\fill}
        \begin{subfigure}[b]{0.3\textwidth}
      	    \includegraphics[scale=0.2]{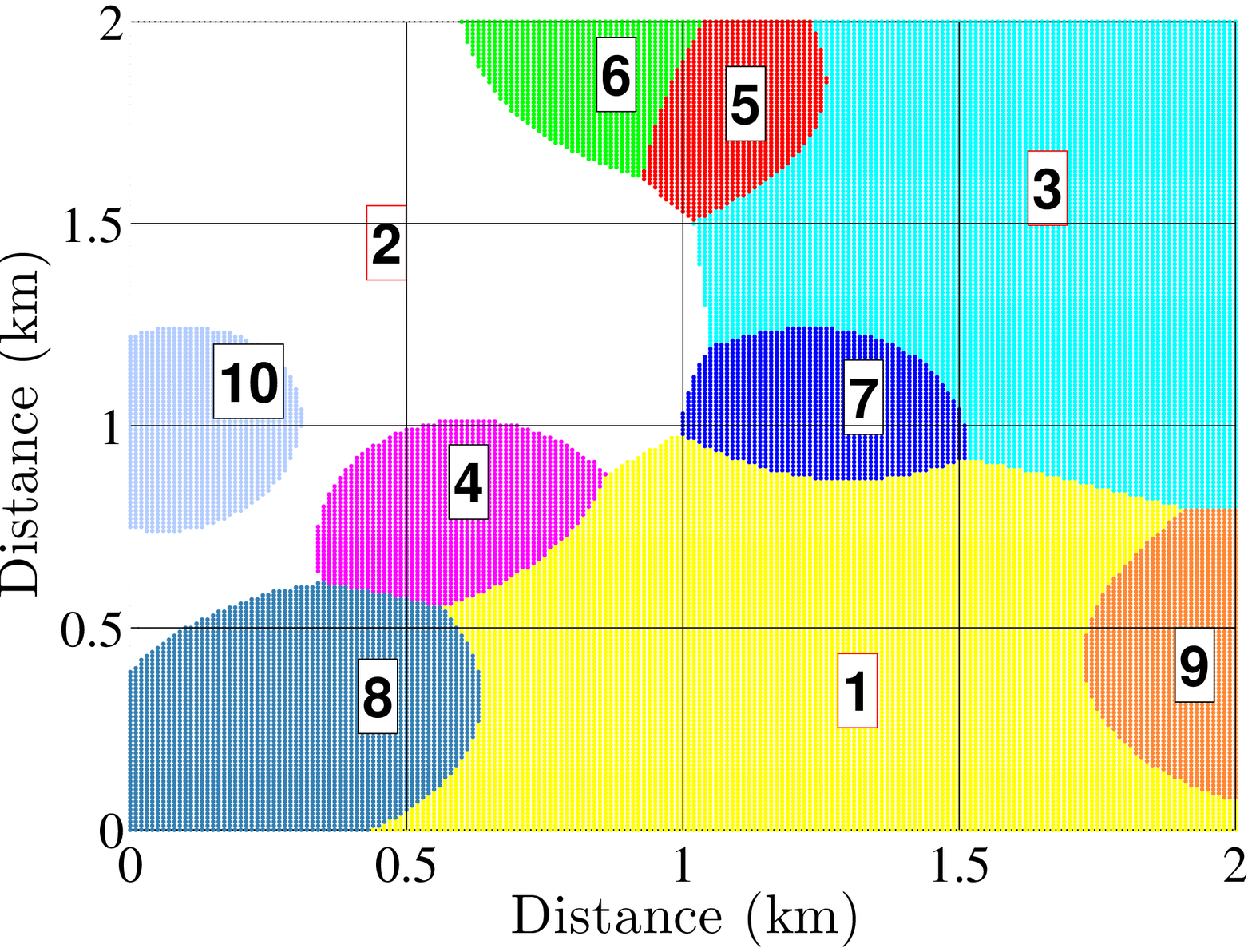}
            \caption{vGALA.}
            \label{fig:sim_6_vgala_area}
       \end{subfigure}\hfill
      \begin{subfigure}[b]{0.3\textwidth}
      	    \includegraphics[scale=0.2]{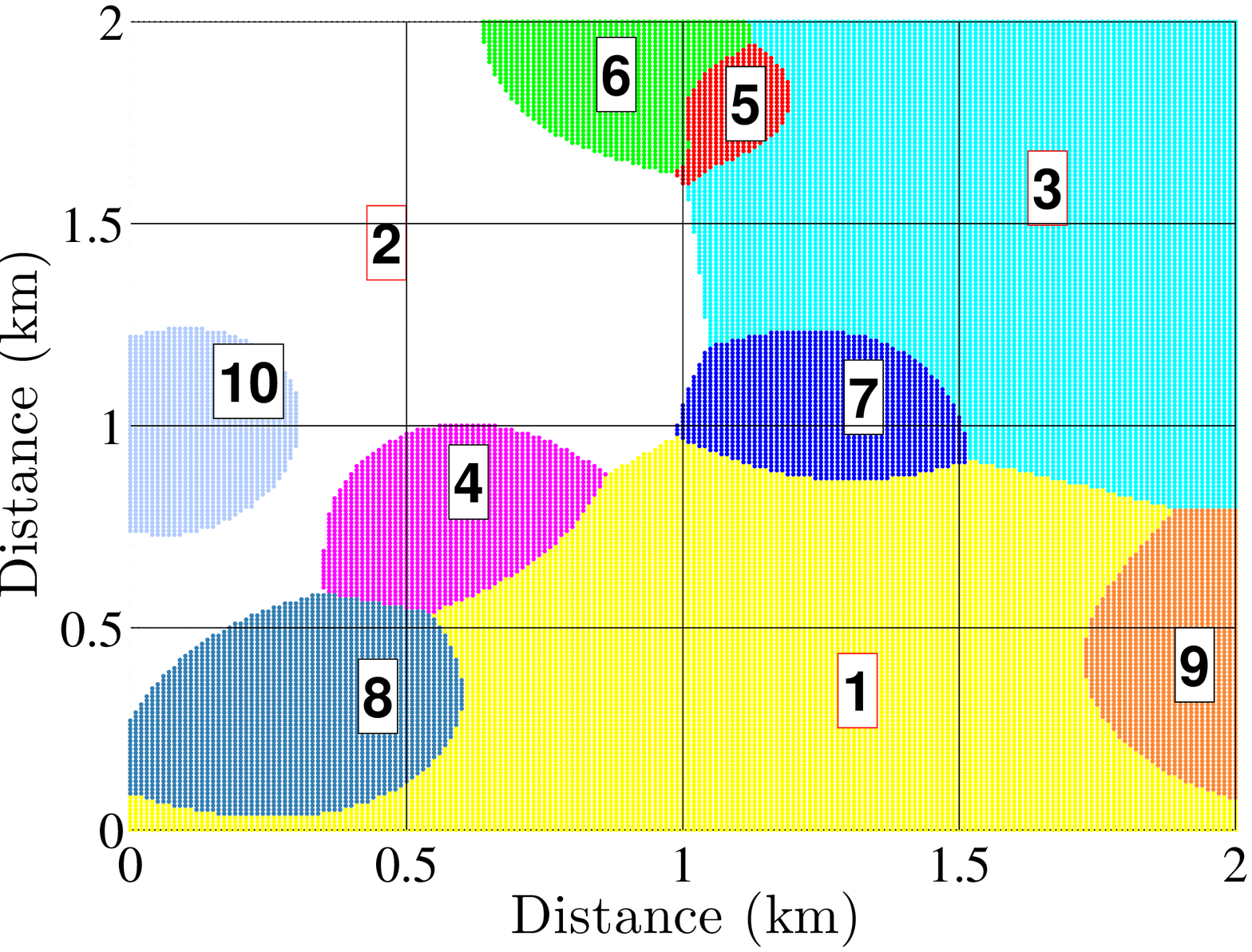}
            \caption{NUA-NC.}
            \label{fig:sim_6_nua_no_cache_area}
       \end{subfigure}\hfill
        \begin{subfigure}[b]{0.3\textwidth}
      	    \includegraphics[scale=0.2]{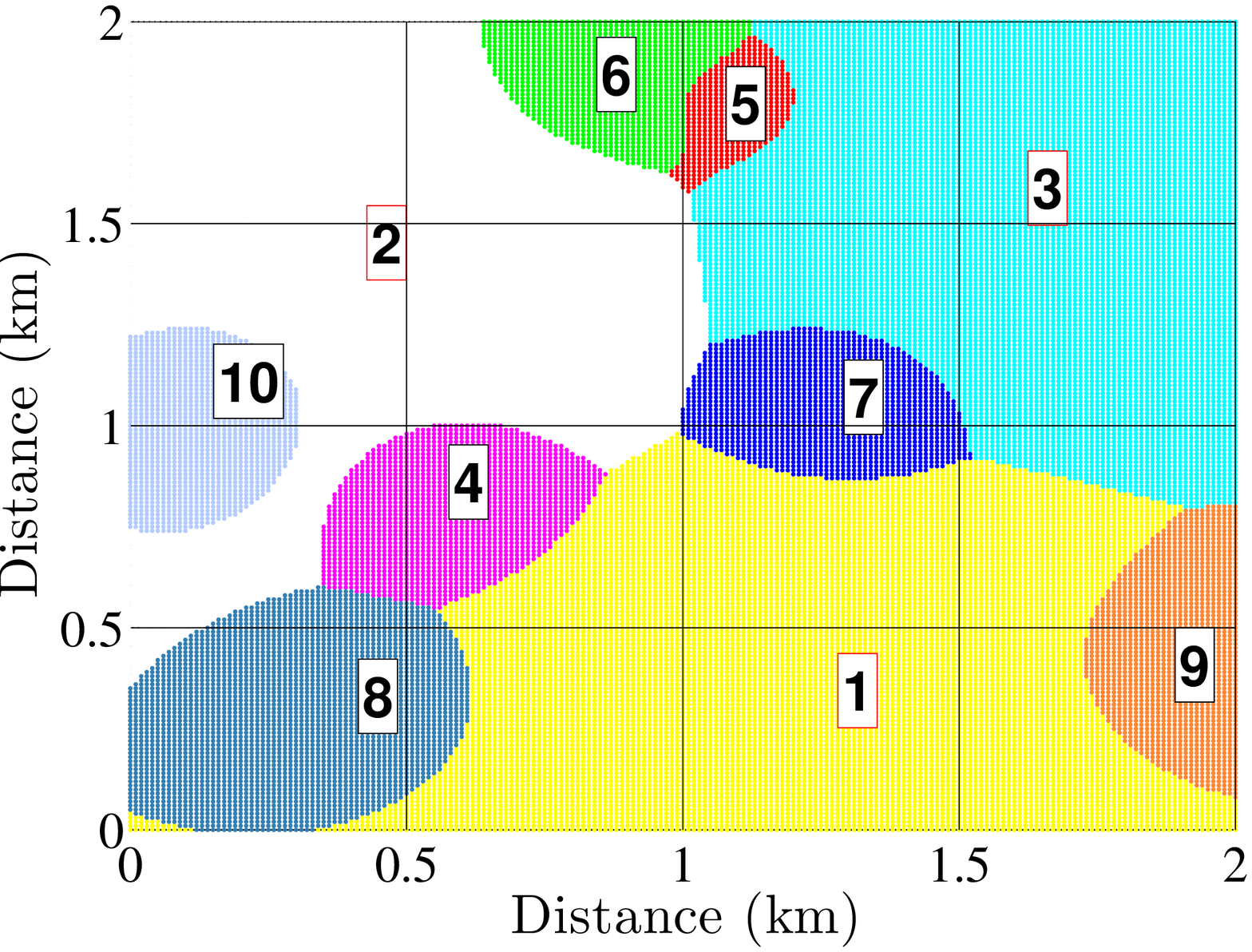}
            \caption{NUA.}
            \label{fig:sim_6_nua_area}
       \end{subfigure}\hfill
    \caption{%
       The coverage areas of different schemes ($R_{5}=1$ $Mbps$ and $R_{4,6-10}=5$ $Mbps$).
     }%
   \label{fig:sim_6_area_comp}
\end{figure*}

Fig. \ref{fig:sim_5_latency} shows the the impact of the cache awareness on the traffic delivery latency. In the simulation, we set $\kappa=0$ for both the NUA scheme and the NUA-NC scheme to focus on the performance of the traffic delivery latency. Thus, both schemes are unaware of the green power utilization. As shown in Fig. Fig. \ref{fig:sim_5_latency}, when the backhaul data rate is small, the cache awareness helps to reduce the traffic delivery latency.

\begin{figure}
\centering
\includegraphics[scale=0.3]{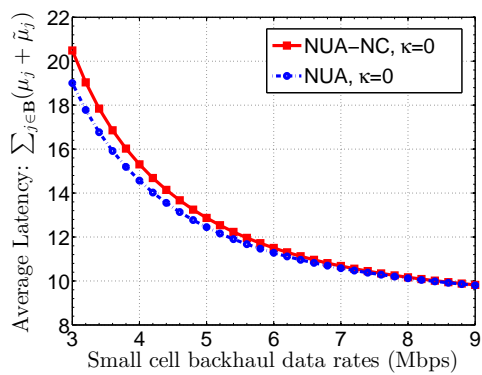}
\caption{The impact of the cache awareness on the traffic delivery latency.}
\label{fig:sim_5_latency}
\end{figure}

\section{Conclusion}
\label{sec:conclusion}
In this paper, we have proposed a network utility aware (NUA) traffic load balancing scheme for backhaul-constrained cache-enabled SCNs with hybrid power supplies. During the procedure of establishing user associations, the NUA traffic load balancing scheme considers four network utilities: green power utilization, the traffic delivery latency in BSs, the traffic delivery latency in backhaul, and the cache hit ratio. By optimizing the user association, the NUA traffic load balancing scheme strikes a tradeoff between the green power utilization and the traffic delivery latency in the network. The NUA traffic load balancing scheme adapts the user association according to the dynamics of green power, BS capacity, backhaul data rates, and the cache hit ratio. It significantly reduces the traffic delivery latency when the network is constrained by the backhaul data rate. Moreover, by adjusting the system parameters, e.g., $\kappa$, the NUA scheme is able to adjust the tradeoff between the brown power consumption and the traffic delivery latency.
\bibliographystyle{IEEEtran}
\bibliography{mybib}

\end{document}